\documentclass[10pt]{article}

\usepackage[lined,boxed,ruled,vlined,linesnumbered,nofillcomment]{algorithm2e}
\usepackage{bm,bbm}
\usepackage{tabularx,ragged2e,booktabs,caption,framed}
\usepackage[utf8]{inputenc}
\usepackage{float,url,hyperref,epic,eepic}
\usepackage{xcolor, enumerate}
\usepackage{graphicx}
\usepackage{amsmath, amsfonts, amssymb, amscd} 
\usepackage{calc}
\usepackage[numbers,square]{natbib}

\usepackage{listings}
\makeatletter
\newcommand{\pushright}[1]{\ifmeasuring@#1\else\omit\hfill$\displaystyle#1$\fi\ignorespaces}
\newcommand{\pushleft}[1]{\ifmeasuring@#1\else\omit$\displaystyle#1$\hfill\fi\ignorespaces}
\makeatother

\DeclareMathOperator{\poly}{poly}

\newcommand{\ttau}{{\tilde{\tau}}}

\DeclareMathOperator{\diag}{diag}
\DeclareMathOperator{\Ring}{R}
\DeclareMathOperator{\cop}{cpr}
\DeclareMathOperator{\sgn}{sgn}
\DeclareMathOperator{\odty}{odt}
\DeclareMathOperator{\sig}{sig}

\DeclareMathOperator{\adj}{adj}
\DeclareMathOperator{\SYM}{sym}
\DeclareMathOperator{\ord}{ord}
\DeclareMathOperator{\can}{can}

\DeclareMathOperator{\tsym}{\textsc{sym}} 

\newcommand{\sympk}{\tsym_{p^k}}

\newcommand{\canp}{\can_p}
\newcommand{\cant}{\can_2}
\newcommand{\sigp}{\sig_p}

\newcommand{\copt}{\cop_2}
\newcommand{\copp}{\cop_p}

\newcommand{\scale}{\mathbb{I}}

\newcommand{\tMQ}{\tilde{\MQ}}
\newcommand{\tMD}{\tilde{\MD}}

\newcommand{\SymT}{\Upsilon_2}

\newcommand{\zRz}[1]{\bbZ/#1\bbZ}
\newcommand{\zqz}{\zRz{q}}
\newcommand{\zpz}{\zRz{p}}

\newcommand{\ztkz}{\zRz{2^k}}
\newcommand{\ztsz}{\zRz{2^s}}
\newcommand{\zpkz}{\zRz{p^k}}

\newcommand{\modpk}{\bmod{p^k}}
\newcommand{\modtk}{\bmod{2^k}}
\newcommand{\eqq}{\overset{q}{\sim}}

\newcommand{\eqp}{\overset{p^*}{\sim}}
\newcommand{\eqpk}{\overset{p^k}{\sim}}
\newcommand{\eqt}{\overset{2^*}{\sim}}
\newcommand{\eqx}[1]{\overset{#1}{\sim}}

\newcommand{\ordp}{\ord_p}
\newcommand{\ordt}{\ord_2}
\newcommand{\sgnp}{\sgn_p}
\newcommand{\sgnt}{\sgn_2}

\DeclareMathOperator{\scalep}{\scale_p}
\DeclareMathOperator{\scalet}{\scale_2}

\newcommand{\MTP}{\mathtt{T}^{+}}
\newcommand{\MTM}{\mathtt{T}^{-}}
\newcommand{\MTT}{\mathtt{T}}
\newcommand{\taup}{\tau^{+}}
\newcommand{\taum}{\tau^{-}}
\newcommand{\totk}{\underset{2^k}{\to}}

\newcommand{\MA}{\mathtt{A}} \newcommand{\MB}{\mathtt{B}}
 \newcommand{\MD}{\mathtt{D}}
 
 \newcommand{\MI}{\mathtt{I}}

\newcommand{\MQ}{\mathtt{Q}} \newcommand{\MS}{\mathtt{S}}
\newcommand{\MT}{\mathtt{T}} \newcommand{\MU}{\mathtt{U}}
\newcommand{\MV}{\mathtt{V}} \newcommand{\MW}{\mathtt{W}}
\newcommand{\MX}{\mathtt{X}} \newcommand{\MY}{\mathtt{Y}}

\newcommand{\Vb}{\mathbf{b}}

\newcommand{\Vx}{\mathbf{x}} 
\newcommand{\Vy}{\mathbf{y}}

\newcommand{\Vv}{\mathbf{v}}
\newcommand{\Vw}{\mathbf{w}}

\DeclareMathOperator{\type}{type}

\DeclareMathOperator{\SGN}{\textsc{sgn}}
\DeclareMathOperator{\SGNI}{\textsc{sgn}^\times}

\DeclareMathOperator{\gl}{GL}
\DeclareMathOperator{\SL}{SL}
\DeclareMathOperator{\gln}{GL_n}

\DeclareMathOperator{\sln}{SL_n}
\DeclareMathOperator{\glt}{GL_2}

\newcommand{\legendre}[2]{%
\left( \frac{#1}{#2} \right)%
}

\newtheorem{fact}{\textsc{Fact}}

\newcommand{\bbR}{\mathbb{R}}

\newcommand{\bbQ}{\mathbb{Q}}
\newcommand{\bbZ}{\mathbb{Z}}
\newcommand{\bbP}{\mathbb{P}}

\iftrue

\fi

\DeclareSymbolFont{bbold}{U}{bbold}{m}{n}
\DeclareSymbolFontAlphabet{\mathbbold}{bbold}

\newcommand{\AND}{\textbf{ and }}

\newsavebox{\theorembox} \newsavebox{\lemmabox}
\newsavebox{\corollarybox} \newsavebox{\propositionbox}
\newsavebox{\examplebox} \newsavebox{\conjecturebox}
\newsavebox{\algbox} \newsavebox{\qbox} \newsavebox{\problembox}
\newsavebox{\definitionbox} \newsavebox{\assumptionbox}
\newsavebox{\hypothesisbox} \newsavebox{\obsbox} 
\savebox{\theorembox}{\noindent\bf Theorem} 
\savebox{\lemmabox}{\noindent\bf Lemma}
\savebox{\corollarybox}{\noindent\bf Corollary}
\savebox{\propositionbox}{\noindent\bf Proposition}
\savebox{\examplebox}{\noindent\bf Example}
\savebox{\conjecturebox}{\noindent\bf Conjecture}
\savebox{\algbox}{\noindent\bf Algorithm}
\savebox{\qbox}{\noindent\bf Question}
\savebox{\definitionbox}{\noindent\bf Definition}
\savebox{\problembox}{\noindent\bf Problem}
\savebox{\assumptionbox}{\noindent\bf Assumption}
\savebox{\hypothesisbox}{\noindent\bf Hypothesis}
\savebox{\obsbox}{\noindent\bf Observation}
\newtheorem{theorem}{\usebox{\theorembox}}
\newtheorem{lemma}[theorem]{\usebox{\lemmabox}}

\newtheorem{definition}{\usebox{\definitionbox}}
\newcommand{\qed}{\;\;\;\Box} 
\newenvironment{proof}{{\bf Proof:}}{\hfill\(\qed\)\newline}

\title{Computing the $p$-adic Canonical Quadratic Form in Polynomial Time}
\author{
	Chandan Dubey\\
	\texttt{chandan.dubey@inf.ethz.ch}
	\and
	Thomas Holenstein\\
	\texttt{thomas.holenstein@inf.ethz.ch}
}
\date{Institut f\"ur Theoretische Informatik, ETH Z\"urich}

\begin{document}
\maketitle

\begin{abstract}
An $n$-ary integral quadratic form is a formal expression
$Q(x_1,\cdots,x_n)=\sum_{1\leq i,j\leq n}a_{ij}x_ix_j$ in $n$-variables 
$x_1,\cdots,x_n$, where $a_{ij}=a_{ji} \in \bbZ$.
We present a randomized polynomial time algorithm
that given a quadratic form $Q(x_1,\cdots,x_n)$, a prime $p$, and a 
positive integer $k$ outputs a $\MU \in \gln(\zpkz)$ such that
$\MU$ transforms $Q$ to its $p$-adic canonical form.
\end{abstract}

\section{Introduction}

Let $\Ring$ be a commutative ring with unity and $\Ring^\times$ 
be the set of units (i.e., invertible elements) of $\Ring$. A 
quadratic form over the ring $\Ring$ in $n$-formal variables 
$x_1,\cdots,x_n$ in an expression 
$\sum_{1\leq i, j \leq n}a_{ij}x_ix_j$, where $a_{ij}=a_{ji} 
\in \Ring$. A quadratic form can equivalently be represented by
a symmetric matrix $\MQ^n=(a_{ij})$ such that 
$Q(x_1,\cdots,x_n)=(x_1,\cdots,x_n)'\MQ(x_1,\cdots,x_n)$. The
quadratic form is called integral if $\Ring=\bbZ$ and the 
determinant of the quadratic form $Q$ is defined as $\det(\MQ)$.
In this paper, we concern ourselves with integral quadratic
forms, henceforth referred only as {\em quadratic forms}.

Let $p$ be a prime and $k$ be a positive integer. Then, two quadratic 
forms $\MQ_1, \MQ_2$ are said to be $p^k$-equivalent if there is a 
$\MU \in \gln(\zpkz)$ such that $\MU'\MQ_1\MU \equiv \MQ_2 \bmod{p^k}$
(denoted $\MQ_1 \eqpk \MQ_2$).
Two quadratic forms are $p^*$-equivalent if they are 
$p^k$-equivalent for every $k>0$ (denoted $\MQ_1 \eqp \MQ_2)$. 
Intuitively, $p^k$-equivalence means
that there exists an invertible linear change of variables over $\zpkz$ that
transforms one form to the other.

For a quadratic form $\MQ$ and prime $p$, the set 
$\{\MS^n \mid \MS \eqp \MQ\}$ is the set of 
$p^*$-equivalent forms of $\MQ$ (also called $p^*$-equivalence class of 
$\MQ$).  In this paper,
we are interested in defining and computing a ``canonical'' 
quadratic form for the $p^*$-equivalence class of a given
quadratic form $\MQ$. In
particular, we are interested in showing the existence of
a function $\canp$ such that for all integral quadratic forms
$\MQ$, $\canp(\MQ) \in \{\MS \mid \MS \eqp \MQ\}$; with the
property that if $\MQ_1 \eqp \MQ_2$ then 
$\canp(\MQ_1) = \canp(\MQ_2)$.
We also consider a related problem of coming up with a 
canonicalization procedure. In particular, we want a polynomial
time algorithm that given $\MQ, p$ and a positive integer $k$,
finds $\MU \in \gln(\zpkz)$ such that 
$\MU'\MQ\MU \equiv \canp(\MQ)\bmod{p^k}$.

It is not difficult to show the existence of a canonical form.
For example, we can go over the $p^*$-equivalence class of
$\MQ$ and output the form which is lexicographically the smallest
one. But, this form gives us no meaningful information about
$\MQ$ or the $p^*$-equivalence class of $\MQ$.

Gauss \cite{Gauss86} gives a complete classification of binary
quadratic forms (i.e., $n=2$). Mathematicians have been more
interested in coming up with a list of necessary and sufficient
conditions for $p^*$-equivalence (also called {\em equivalence over the
$p$-adic integers} $\bbZ_p$). There are several competing but equivalent 
candidates for the set of conditions
\cite{Cassels78, OMeara73, CS99, Kitaoka99}. We choose to
use the same set of conditions as Conway-Sloane \cite{CS99},
called the $p$-{\em symbol} of a quadratic form.


For odd prime $p$,
the $p$-canonical form is implicit in Conway-Sloane 
\cite{CS99} and is also described explicitly by Hartung
\cite{Hartung08, Cassels78}. The canonicalization algorithm in this case
is not complicated and can be claimed to be implicit in 
Cassels \cite{Cassels78}.

The definition of canonical form for the case of prime~2 is 
quite involved and needs careful analysis\footnote{
Cassels (page~117, Section~4, \cite{Cassels78}),
referring to the canonical forms for $p=2$
observes that ``only a masochist is invited to read the rest''.} 
\cite{Jones44, CS99, Watson60}. 
Jones \cite{Jones44} presents the most complete description
of the $2$-canonical form. His method is to come up with a small
$2$-canonical forms and then showing that every quadratic form is
$2^*$-equivalent to one of these. Unfortunately, a few of his 
transformations are existential i.e., he shows that a transformations
with certain properties exists without explicitly finding them.

Conway-Sloane \cite{CS99}, instead, compute a description of
a quadratic form (called canonical $2$-symbol) with the property 
that two quadratic forms
are $2^*$-equivalent iff they have the same canonical $2$-symbol. They
do not provide a $2$-canonical form i.e., 
$\cant(\MQ) \in \{\MS \mid \MS \eqt \MQ\}$.

The $p$-canonical forms are very useful in the study of quadratic 
forms and their equivalence, 
see \cite{Siegel35, Jones44, Jones50, Cassels78, CS99, Kitaoka99, Hartung08}.

\paragraph{Our Contribution.} We give polynomial time
$p$-canonicalization algorithm. In particular, we present
an algorithm that given an 
integral quadratic form $\MQ$ runs in time
$\poly(n, \log \det(\MQ))$ and 
outputs $\can_p(\MQ)$. 

Given an integral quadratic form $\MQ$, a positive integer $k$ 
and a prime $p$, we also provide a randomized
$\poly(n,\log \det(\MQ), \log p, k)$ algorithm that outputs a matrix
$\MU \in \gln(\zpkz)$ such that $\MU'\MQ\MU \equiv \canp(\MQ) \pmod{p^k}$.
This algorithm is especially useful if we want to find a transformation
that maps $\MQ_1$ to $\MQ_2$ over $\zpkz$, where 
$\MQ_1 \eqx{p^k}\MQ_2$. In this case, the required transformation is
$\MU_1\MU_2^{-1} \bmod {p^k}$, where 
$\MU_1'\MQ_1\MU_1 \equiv \canp(\MQ_1)\equiv\MU_2'\MQ_2\MU_2 \pmod{p^k}$.

\section{Preliminaries}

Integers and ring elements are denoted by lowercase letters,
vectors by bold lowercase letters and matrices by typewriter
uppercase letters. The $i$'th component of a vector $\Vv$ is
denoted by $v_i$. We use the notation $(v_1,\cdots,v_n)$ for
a column vector and the transpose of matrix $\MA$ is denoted by
$\MA'$. The matrix
$\MA^n$ will denote a $n\times n$ square matrix.
The scalar product of two vectors will be denoted 
$\Vv'\Vw$ and equals $\sum_i v_iw_i$. The standard Euclidean norm of the 
vector $\Vv$ is denoted by $||\Vv||$ and equals $\sqrt{\Vv'\Vv}$.

If $\MQ_1^n, \MQ_2^m$ are matrices, then the {\em direct product}
of $\MQ_1$ and $\MQ_2$ is denoted by $\MQ_1\oplus\MQ_2$ and is
defined as $\diag(\MQ_1,\MQ_2)=\begin{pmatrix}
\MQ_1&0\\0&\MQ_2
\end{pmatrix}$.
Given two matrices $\MQ_1$ and $\MQ_2$ with the same number
of rows, $[\MQ_1,\MQ_2]$ is the matrix which is obtained
by concatenating the two matrices columnwise.
A matrix is called unimodular
if it is an integer $n\times n$ matrix with determinant $\pm 1$.
If $\MQ^n$ is a $n\times n$ integer matrix and $q$ is a positive integer
then $\MQ \bmod{q}$ is defined as the matrix with all entries of $\MQ$
reduced modulo $q$.


Let $\Ring$ be a commutative ring with unity and $\Ring^\times$ be the
set of units (i.e., invertible elements) of $\Ring$.
If $\MQ \in \Ring^{n\times n}$ is a square matrix, the 
{\em adjugate} of $\MQ$ is defined as the transpose of 
the cofactor matrix and is denoted by $\adj(\MQ)$. The matrix $\MQ$ 
is invertible if and only if $\det(\MQ)$ is a unit of $\Ring$. 
In this case, $\adj(\MQ)=\det(\MQ)\MQ^{-1}$. The set of invertible 
$n\times n$ matrices over $\Ring$ is denoted by $\gln(\Ring)$. The 
subset of matrices with determinant $1$ will be denoted by 
$\sln(\Ring)$.

\begin{fact}\label{fact:gln}
A matrix $\MU$ is in $\gln(\Ring)$ iff $\det(\MU) \in \Ring^\times$. 
\end{fact}

The set of odd primes is denoted by $\bbP$. 
We define $\bbQ/(-1)\bbQ=\zRz{(-1)}:=\bbR$. 
For every prime $p$ and positive integer $k$,
we define the ring $\zpkz=\{0,\cdots,p^k-1\}$, where product
and addition is defined modulo $p^k$.

Let $p$ be a prime, and $a, b$ be integers. Then, 
$\ordp(a)$ is the largest integer exponent of $p$ such that $p^{\ordp(a)}$
divides $a$. We let $\ordp(0) = \infty$. The $p$-coprime part of $a$ is then 
$\copp(a)=\frac{a}{p^{\ordp(a)}}$. Note that $\copp(a)$ is,
by definition, a unit of $\zpz$. 
For $\frac{a}{b}$, a rational number, we define
$\ordp(\frac{a}{b})=\ordp(a)-\ordp(b)$. The $p$-coprime part of $\frac{a}{b}$
is denoted as
$\copp(\frac{a}{a})$ and equals $\frac{a/p^{\ordp(a)}}{b/p^{\ordp(b)}}$.
For a positive integer $q$,
one writes $a\equiv b \bmod{q}$, if $q$ 
divides $a-b$. By $x:=a \bmod{q}$, we mean that
$x$ is assigned the unique value $b \in \{0,\cdots,q-1\}$ such that 
$b \equiv a \bmod{q}$.
An integer $t$ is 
called a {\em quadratic residue} modulo $q$ if $\gcd(t,q)=1$ and
$x^2\equiv t \bmod{q}$ has a solution. 

\begin{definition}\label{def:Legendre} 
Let $p$ be an odd prime, and $t$ be a positive integer with
$\gcd(t,p)=1$.
Then, the Legendre-symbol of $t$ with 
respect to $p$ is defined as follows.
\begin{displaymath}
\legendre{t}{p} = t^{(p-1)/2} \bmod p = \left\{ \begin{array}{ll}
1 & \textrm{if $t$ is a quadratic residue modulo $p$}\\
-1 & \textrm{otherwise.}
\end{array}\right.
\end{displaymath}
\end{definition}

For the prime~2, there is an extension of Legendre symbol called the
Kronecker symbol. It is defined for odd integers $t$ and 
$\legendre{t}{2}$ equals $1$ if $t\equiv \pm 1 \bmod 8$, and $-1$
if $t \equiv \pm 3 \bmod 8$.

The $p$-sign of $t$, denoted $\sgnp(t)$, is defined as 
$\legendre{\copp(t)}{p}$
for odd primes $p$ and $\copt(t) \bmod 8$ otherwise. We
also define $\sgnp(0)=0$, for all primes $p$.
Thus,
\[
\sgnp(0) = 0 \qquad 
\sgnp(t>0) \in \left\{
	\begin{array}{ll}
	\{+1,-1\} & \text{if $p$ is odd}\\
	\{1,3,5,7\} & \text{otherwise}
	\end{array}\right.
\]

The following lemma is well known.

\begin{lemma}\label{lem:QR}
Let $p$ be an odd prime. Then, there
are $\frac{p-1}{2}$ quadratic residues and $\frac{p-1}{2}$ 
quadratic non-residues modulo $p$. Also, every quadratic residue
in $\zpz$ can be written as a sum of two quadratic non-residues
and every quadratic non-residue can be written as a sum of two quadratic
residues.
\end{lemma}

An integer $t$ is a square modulo $q$ 
if there exists an integer $x$ such that $x^2\equiv t \pmod{q}$.
The integer $x$ is called the {\em square root} of $t$
modulo $q$. If no such $x$ exists, then $t$ is a non-square modulo $q$.

The following lemma is folklore and gives the necessary and sufficient
conditions for an integer $t$ to be a square modulo $p^k$. For
completeness, a proof is provided in Appendix \ref{sec:Proofs}.

\begin{lemma}\label{lem:Square}
Let $p$ be a prime, $k$ be a positive integer and 
$t \in \zpkz$ be a non-zero integer. Then, $t$ is a 
square modulo $p^k$ if and only if $\ordp(t)$ is even and
$\sgnp(t)=1$. 
\end{lemma}

\begin{definition}\label{def:QNR}
Let $p$ be an odd prime. Then, $\sigma_p$ is the smallest quadratic
non-residue modulo $p$.
\end{definition}

Assuming GRH, $\sigma_p$ is
a number less than $3(\ln p)^2/2$ \cite{Ankeny52,Wedeniwski01}
and hence can be found deterministically in $O(\log^3 p)$ ring operations
over $\zpz$.

\begin{definition}\label{def:NonSquare}
Let $p$ be a prime and $\frac{x}{y}$ be a rational number. Then, $\frac{x}{y}$
can be uniquely written as $\frac{x}{y}=p^{\alpha}\frac{a}{b}$, 
where $a,b$ are  units of $\zpz$. We say that $\frac{x}{y}$ is a 
$p$-antisquare if $\alpha$ is odd and $\sgnp(a) \neq \sgnp(b)$.
\end{definition}

For convenience we define integers $k_p$, and a
{\em completion} of an integer $q$ (denoted $\overline{q}$),
as follows.
\begin{align}\label{not:kp}
k_p = \left\{\begin{array}{ll} 3 & \text{if $p=2$, and}\\
1 & \text{$p$ odd prime.}
\end{array}\right.
\qquad 
\overline{q}= q\prod_{p|2q}p^{k_p}
\end{align}

We also introduce the following notations.
\begin{align*}
\SGNI &= \{1,3,5,7\} \\
\MQ_1 \underset{\MU, p^k}{\to} \MQ_2 & \text{ denotes } 
\MQ_2 \equiv \MU'\MQ_1\MU \bmod{p^k} \\
\MTP = \begin{pmatrix}2&1\\1&4\end{pmatrix}
&\qquad
\MTM = \begin{pmatrix}2&1\\1&2\end{pmatrix}\\
\MTT &\in \{\MTM, \MTP\}
\end{align*}

\begin{definition}\label{def:Prim}
Let $p^k$ be a prime power. A vector $\Vv \in (\zpkz)^n$ is called 
primitive if there exists a component $v_i$, $i \in [n]$, of $\Vv$ such
that $\gcd(v_i,p)=1$. Otherwise, the vector $\Vv$ is 
non-primitive.
\end{definition}

Our definition of primitiveness of a vector is different but equivalent 
to the 
usual one in the literature. A
vector $\Vv \in (\zqz)^n$ is called primitive over $\zqz$ for a 
composite integer $q$ if it is primitive modulo $p^{\ordp(q)}$ for
all primes that divide $q$.

\begin{lemma}\label{lem:ExtendPrimitive}
Let $p$ be a prime, $k$ be a positive integer and 
$\Vx \in (\zpkz)^n$ be a primitive vector.
Then, an $\MA$ can be found in $O(n^2)$ ring operation
such that $[\Vx,\MA] \in \sln(\zpkz)$.
\end{lemma}
\begin{proof}
The column vector $\Vx=(x_1,\cdots,x_n)$ is primitive, hence there 
exists a $x_i$, $i\in [n]$ such that $x_i$ is invertible over $\zpkz$.
It is easier to write the matrix $\MU$, which equal $[\Vx,\MA]$ where
the row $i$ and $1$ or $[\Vx,\MA]$ are swapped.
\[
\MU = \begin{pmatrix}x_i & 0 \\
\Vx_{-i} & x_i^{-1}\bmod{p^k} \oplus \MI^{n-2}
\end{pmatrix}\qquad \Vx_{-i}=(x_1,\cdots,x_{i-1},x_{i+1},\cdots,x_n)
\]
The matrix $\MU$ has determinant~1 modulo $p^k$ and hence is invertible
over $\zpkz$. The lemma now follows from the fact that the swapped matrix 
is invertible iff the original matrix is invertible.
\end{proof}

\paragraph{Quadratic Form.} 
An $n$-ary quadratic form over a ring $\Ring$
is a symmetric matrix $\MQ \in \Ring^{n\times n}$, interpreted as the following
polynomial in $n$ formal variables $x_1,\cdots, x_n$ of uniform degree~2.
\[
\sum_{1\leq i,j \leq n}\MQ_{ij}x_ix_j = 
\MQ_{11}x_1^2 + \MQ_{12}x_1x_2 + \cdots = \Vx'\MQ\Vx
\]
The quadratic form is called {\em integral} if it is defined over the ring
$\bbZ$. It is called positive definite if for all non-zero column vectors
$\Vx$, $\Vx'\MQ\Vx > 0$. This work deals with integral quadratic forms,
henceforth called simply {\em quadratic forms}.
The {\em determinant} of the quadratic 
form is defined as $\det(\MQ)$. 
A quadratic form is called {\em diagonal} if $\MQ$ is a diagonal matrix. 

Given a set of formal variables 
$\Vx=\begin{pmatrix}x_1 & \cdots & x_n\end{pmatrix}'$ one can make a linear 
change of variables to $\Vy=\begin{pmatrix}y_1 & \cdots & y_n\end{pmatrix}'$ 
using a matrix $\MU \in \Ring^{n\times n}$ by setting $\Vy=\MU\Vx$. 
If additionally, 
$\MU$ is invertible over $\Ring$ i.e., $\MU \in \gln(\Ring)$, then this 
change of 
variables is reversible over the ring. We now define the equivalence of
quadratic forms over the ring $\Ring$ (compare with Lattice Isomorphism).

\begin{definition}\label{def:equiv}
Let $\MQ_1^n, \MQ_2^n$ be quadratic forms over a ring $\Ring$. They are called 
$\Ring$-{\em equivalent} if there exists a $\MU \in \gln(\Ring)$ such that 
$\MQ_2=\MU'\MQ_1\MU$.
\end{definition}

If $\Ring=\zqz$, for some positive integer $q$, then two integral
quadratic forms $\MQ_1^n$ and $\MQ_2^n$ will be called $q$-equivalent (denoted,
$\MQ_1\eqq \MQ_2$)
if there exists a matrix $\MU \in \gln(\zqz)$ such that 
$\MQ_2\equiv\MU'\MQ_1\MU\pmod q$.
For a prime $p$, they
are $p^*$-equivalent (denoted, $\MQ_1\eqp \MQ_2$) if  they
are $p^k$-equivalent for every positive integer $k$. Additionally,
$(-1)^*$-equivalence as well as $(-1)$-equivalence mean
equivalence over the reals $\bbR$.

Let $\MQ^n$ be a $n$-ary integral quadratic form, and $q,t $ be positive 
integers. If the equation $\Vx'\MQ\Vx \equiv t \pmod{q}$ has a solution
then we say that $t$ has a $q$-representation in $\MQ$ (or $t$ has 
a representation in
$\MQ$ over $\zqz$). Solutions 
$\Vx \in (\zqz)^n$ to the equation are called $q$-{\em representations} of
$t$ in $\MQ$. We classify the representations into two
categories: {\em primitive} and {\em non-primitive} (see Definition 
\ref{def:Prim}).

For the following result, see Theorem~2, \cite{Jones50}.

\begin{theorem}\label{thm:diagonal}
An integral quadratic form $\MQ^n$ is equivalent to a quadratic
form 
$q_1 \oplus \cdots \oplus q_a \oplus q_{a+1} \oplus \cdots \oplus q_n$ 
over the field of rationals $\bbQ$,
where $a \in [n]$, $q_1, \cdots, q_a$ are positive rational numbers and
$q_{a+1}, \cdots, q_n$ are negative rational numbers.
\end{theorem}

The {\em signature} (also, $(-1)$-signature) of the form 
$\MQ$ (denoted $\sig(\MQ)$, also  $\sig_{(-1)}(\MQ)$) is defined
as the number $2a-n$, where $a$ is the integer in 
Theorem \ref{thm:diagonal}. 

Each rational number $q_i$ in Theorem
\ref{thm:diagonal} can be written uniquely as 
$p^{\alpha_i}a_i$, where $\alpha_i = \ordp(q_i)$ and $a_i = \copp(q_i)$.
Let $m$ be the number of $p$-antisquares among $q_1,\cdots,q_n$. 
Then, we define the $p$-signature of $\MQ$ as follows.
\begin{align}\label{def:PSignature}
\sigp(\MQ)=\left\{
\begin{array}{ll}
p^{\alpha_1}+ p^{\alpha_2} + \cdots + p^{\alpha_n} + 4m  \pmod{8}  & p\neq 2 \\
a_1 + a_2 + \cdots + + a_n + 4m \pmod{8} & p=2 
\end{array}\right.
\end{align}
The $2$-signature is also known as the {\em oddity} and is denoted by
$\odty(\MQ)$. 
Even though there are different ways to diagonalize a quadratic form
over $\bbQ$, the signatures are an invariant for the quadratic form.

\paragraph{Randomized Algorithms.} 
Our randomized algorithms are 
Las Vegas algorithms. They either fail
and output nothing, or produce a correct answer. The 
probability of failure is bounded by a constant. Thus, for any
$\delta>0$, it is possible to repeat the algorithm 
$O(\log \frac{1}{\delta})$ times and succeed with probability at least
$1-\delta$. Henceforth, these algorithms will be called
{\em randomized algorithms}.

Our algorithms perform two kinds of operations. Ring operations e.g.,
multiplication, additions, inversions over $\zpkz$ and operations
over integers $\bbZ$ e.g., multiplications, additions, divisions etc.
The runtime for all these operations is treated as constant i.e., $O(1)$
and the time complexity of the algorithms is measured in terms of 
ring operations. 
For computing a uniform representation, we also need to sample a 
uniform ring element from $\zpkz$. We adapt the convention that sampling
a uniform ring elements also takes $O(1)$ ring operations.
For example, the Legendre symbol of an integer $a$ can be computed
by fast exponentiation in $O(\log p)$ ring operations over $\zpz$ 
while $\ordp(t)$
for $t \in \zpkz$ can be computed by fast exponentiation in $O(\log k)$ ring
operations over $\zpkz$.

Let $\omega$ be the constant, such that multiplying two $n\times n$ 
matrices over $\zpkz$ takes
$O(n^{\omega})$ ring operations. 

\paragraph{Diagonalizing a Quadratic Form.}

For the ring $\zpkz$ such that $p$ is odd, there
always exists an equivalent quadratic form which is also diagonal 
(see \cite{CS99}, Theorem 2, page 369).
Additionally, one can explicitly find the invertible change of
variables that turns it into a diagonal quadratic form.
The situation is tricky over the ring $\ztkz$. Here, it might not
be possible to eliminate all mixed terms, i.e., terms of the form
$2a_{ij}x_ix_j$ with $i\neq j$. For example, consider
the quadratic form $2xy$ i.e., 
$\begin{pmatrix}0 & 1 \\ 1 & 0 \end{pmatrix}$ over $\ztkz$. 
An invertible linear change of variables over $\ztkz$ is
of the following form.
\begin{gather*}
\begin{array}{l}x \to a_1x_1 + a_2x_2\\ 
y \to b_1x_1+b_2x_2\end{array} \qquad 
\begin{pmatrix}a_1 & a_2 \\ b_1 & b_2\end{pmatrix} \text{ invertible over
$\ztkz$}
\end{gather*}
The mixed term after this transformation is $2(a_1b_2+a_2b_1)$. 
As $a_1b_2+a_2b_1 \bmod 2$ is the same as the determinant
of the change of variables above i.e., $a_1b_2-a_2b_1$
modulo $2$; it is not possible for a transformation in 
$\gl_2(\ztkz)$ to eliminate the mixed term. 
Instead, one can show that over $\ztkz$ it is possible to get an 
equivalent form where the mixed terms are disjoint i.e., both
$x_ix_j$ and $x_ix_k$ do not appear, where $i,j,k$ are 
pairwise distinct.
One captures this form by the following definition.

\begin{definition}\label{def:BlockDiagonal}
A matrix $\MD^n$ over integers is in a block diagonal form if it is a
direct sum of type I and type II forms; where type I form is an integer
while type II is a matrix of the form 
$\begin{pmatrix}2^{\ell+1}a & 2^{\ell} b \\ 2^{\ell}b & 2^{\ell+1}c\end{pmatrix}$
with $b$ odd.
\end{definition}

The following theorem is folklore and is also implicit in the proof of
Theorem~2 on page~369 in \cite{CS99}. For completeness, we provide a 
proof in Appendix \ref{sec:BlockDiagonal}.

\begin{theorem}\label{thm:BlockDiagonal}
Let $\MQ^n$ be an integral quadratic form, $p$ be a prime, 
and $k$ be a positive integer. 
Then, there is an algorithm that performs $O(n^{1+\omega}\log k)$ ring operations
and produces a matrix $\MU \in \sln(\zpkz)$ such that $\MU'\MQ\MU\pmod{p^k}$, is a
diagonal matrix for odd primes $p$ and a block diagonal matrix
(in the sense of Definition~\ref{def:BlockDiagonal}) for $p=2$.
\end{theorem}

\begin{definition}\label{def:LocalTransformation}
Let $\MD^n=\oplus_i\MD_i^{n_i}$ be a block diagonal quadratic form 
(Definition \ref{def:BlockDiagonal}). A local transformation
is a matrix $\MU \in \gln(\zpkz)$ which applies a sub-transformation
$\MV \in \gl_a(\zpkz)$ on a contiguous sequence of blocks
$\MB^a=\MD_j\oplus\MD_{j+1}\oplus \cdots$ turning it into 
$\MV'\MB\MV \bmod{p^k}$, 
leaving rest of the blocks in $\MD$ unchanged. 
A local transformation transforms a block diagonal form to a 
$p^k$-equivalent quadratic form.
Given a matrix $\MV \in \gl_a(\zpkz)$ and a contiguous sequence
$\MB^a=\MD_j\oplus\MD_{j+1}\oplus \cdots$ of blocks, to apply
$\MV$ on; the local transformation 
$\MU \in \gln(\zpkz)$ is given by $
\MU = \MI^{n_1+\cdots+n_{j-1}}\oplus \MV \oplus \MI^{n_{j+2}+\cdots}$.
\end{definition}

\subsection{Canonical Blocks}\label{sec:CanBlock}

In this section, we describe the canonical form for a single
Type I and a single Type II block.
For convenience, we introduce the following Type II matrices.

\begin{definition}\label{def:ordt}
Let $\MB=2^\ell\begin{pmatrix}2a & b \\ b & 2c\end{pmatrix}$ be a 
Type II block with $b$ odd. Then, $\ordt(\MB)=\ell$.
\end{definition}

Next, we define something called the $p^k$-symbol of an integer.

\begin{definition}\label{def:SymbolOneDim}
The $p^k$-{\em symbol} of an integer $t$ is
$\sympk(t)=(\ordp(t \modpk), \sgnp(t \modpk))$.
\end{definition}

The next lemma shows the importance of the $p^k$-symbol.

\begin{lemma}\label{lem:IntegerSymbolIsInvariant}
For integers $a,b$ and prime $p$:
$b \eqpk a$ iff $\sympk(a)=\sympk(b)$.
\end{lemma}
\begin{proof}
The lemma is true if $\ordp(a)$ or $\ordp(b)$ is at least
$k$. Hence, we assume that $\ordp(a), \ordp(b) < k$.

We first show that $b \eqpk a$ implies $\sympk(a)=\sympk(b)$.
If $b \eqpk a$ then there 
exists a $u \in (\zpkz)^\times$ such that 
$b \equiv u^2a \pmod{p^k}$. But, multiplying by a square of a unit does
not change the sign i.e., 
$\sgnp(a \bmod{p^k})=\sgnp(u^2a \bmod p^k)=\sgnp(b \bmod{p^k})$.
Also, $\ordp(u)=0$ implies that $\ordp(a)=\ordp(b)$. This shows
that $\sympk(a)=\sympk(b)$.

We now show the converse. Suppose $a$ and $b$ be such that 
$\sympk(a)=\sympk(b)$. Let $\ordp(a)=\ordp(b)=\alpha$. By
definition of $p^k$-symbol, $\sgnp(a \modpk)=\sgnp(b \modpk)$. But then,
\begin{align*}
\sgnp\left(\copp(a) \bmod{p^{k-\alpha}}\right) = 
\sgnp\left(\copp(b) \bmod{p^{k-\alpha}}\right)\\
\iff \sgnp\left(\copp(a)\copp(b)^{-1} \bmod{p^{k-\alpha}}\right) = 1
\end{align*}
By Lemma \ref{lem:Square}, 
$\copp(a)\copp(b)^{-1} \bmod{p^{k-\alpha}}$ is a 
quadratic residue modulo $p^{k-\alpha}$. But then,
there exists a unit $u$ such that 
\[
u^2 \equiv \copp(a)\copp(b)^{-1} \pmod{p^{k-\alpha}}\;.
\] 
Multiplying this equation by 
$\copp(b)p^\alpha$ yields
$u^2b \equiv a \modpk$ or $b \eqpk a$.
\end{proof}

Let $p$ be a prime and $\MB$ be a single block, according to the
Definition \ref{def:BlockDiagonal}. If $\MB$ is of Type I then
$\MB$ is an integer and the canonical function $\canp(\MB)$
is defined as follows.
\[
\canp(\MB) = \left\{
	\begin{array}{ll}
	p^{\ordp(\MB)} & \text{if } p \text{ odd, }\legendre{\copp(\MB)}{p}=1 \\
	p^{\ordp(\MB)}\sigma_p & \text{if } p \text{ odd, }\legendre{\copp(\MB)}{p}=-1 \\
	2^{\ordt(\MB)}(\copt(\MB) \bmod 8) & \text{if }p=2
	\end{array}\right.
\]
The uniqueness of the canonical form follows from 
Lemma \ref{lem:IntegerSymbolIsInvariant}. Otherwise, $\MB$ is a Type II
block and $p=2$. Let
$\MB=2^\ell\begin{pmatrix}2a&b\\b&2c\end{pmatrix}$,
$b$ odd. The square of an odd integer is always equal to~1 modulo $8$.
But then, the quantity $4ac-b^2 \bmod 8 \in \{3,7\}$. The $2$-canonical
form for a Type II block $\MB$ is defined as follows.
\[
\cant\left(2^\ell\begin{pmatrix}2a&b\\b&2c\end{pmatrix}\right) = \left\{
	\begin{array}{ll}
	2^\ell \MTM & 4ac-b^2\equiv 3 \bmod 8 \\
	2^\ell \MTP & 4ac-b^2\equiv 7 \bmod 8
	\end{array}\right.
\]
The uniqueness follows from Lemma~6, \cite{Jones42}.

\subsection{Primitive Representations Modulo $p^k$}

The following theorem gives an algorithmic handle on the 
question of deciding if an integer $t$ has a primitive
$p^*$-representation in $\MQ^n$. The theorem is implicit
in Siegel \cite{Siegel35}.

\begin{theorem}\label{thm:Siegel13}
Let $\MQ^n$ be an integral quadratic form, $t$ be an integer,
$p$ be a prime and $k=\max\{\ordp(\MQ),\ordp(t)\}+k_p$. Then,
if $t$ has a primitive $p^k$-representation in $\MQ$ then 
$t$ has a primitive $p^*$-representation in $\tMQ$ for all
$\tMQ \eqp \MQ$.
\end{theorem}
\begin{proof}
We do the proof in two steps: (i) if $t$ has a primitive $p^k$-
representation in $\MQ$ then $t$ has a primitive $p^*$-representation
in $\MQ$, and (ii) if $t$ has a primitive $p^*$-representation in
$\MQ$ then $t$ has a primitive $p^*$-representation in $\tMQ$
for all $\tMQ$ such that $\tMQ \eqp \MQ$.

The proof of (i) follows.
By assumption, there exists a primitive $\Vx \in (\bbZ/p^k\bbZ)^n$ 
such that $\Vx'\MQ\Vx\equiv t \pmod{p^k}$. Let $a=\Vx'\MQ\Vx$ be 
an integer, then by definition of symbols $a$ and $t$ have the 
same $p^k$-symbol. This implies that for all 
$i\geq k$ there exists a unit $u_i \in \bbZ/p^i\bbZ$ such that 
$u_i^2a\equiv t \pmod{p^i}$. It follows that $u_i\Vx$ is a 
primitive representation of $t$ in $\bbZ/p^i\bbZ$. But,
if $\Vx$ is a primitive representation of $t$ by $\MQ$ over 
$\bbZ/p^i\bbZ$ then $\Vx$ is also a primitive representation 
of $t$ by $\MQ$ over $\bbZ/p^j\bbZ$, for all positive integers 
$j \leq i$. This completes the proof of (i).

The proof of (ii) follows. Let $K$ be an arbitrary positive 
integer and $\Vx \in (\bbZ/p^K\bbZ)^n$ be a primitive vector 
such that $\Vx'\MQ\Vx\equiv t \bmod{p^K}$. As $\tMQ \eqp \MQ$, there exists
$\MU \in \gln(\bbZ/p^K\bbZ)$ such that $\MQ\equiv\MU'\tMQ\MU \bmod{p^K}$.
Thus, $(\MU\Vx)'\tMQ(\MU\Vx) \equiv t \bmod{p^K}$ and $\MU\Vx$
is a $p^K$-representation of $t$ in $\tMQ$. If $\Vx$ is primitive then
so is $\MU\Vx$. As $K$ is arbitrary, the proof of (ii) and hence
the theorem is complete.
\end{proof}

Next, we give several results from \cite{DH14}.
This paper deals with the following problem. Given a quadratic
form $\MQ$ in $n$-variables, a prime $p$, and integers $k,t$
find a solution of $\Vx'\MQ\Vx\equiv t \bmod{p^k}$, if it exists.
Note that it is easy (i.e., polynomial time tester exists) to test 
if $t$ has a $p^k$-representation in $\MQ$.

\begin{theorem}\label{thm:PolyRep}
Let $\MQ^n$ be an integral quadratic form, $p$ be a prime, 
$k$ be a positive integer, $t$ be an element of $\zpkz$.
Then, there is a 
polynomial time Las Vegas algorithm
that performs
$O(n^{1+\omega}\log k+nk^3+n\log p)$ ring operations over 
$\zpkz$ and fails with constant probability (say, at most $\frac1{3}$).
Otherwise, the algorithm
outputs a primitive
 $p^k$-representation of $t$ by $\MQ$,
if such a representation exists. The time complexity can be improved
for the following special cases.
\begin{align*}
\begin{array}{ll}
\text{Type I, $p$ odd} & ~~O(\log k+\log p)\\
\text{Type I, $p=2$} & ~~O(k)\\
\text{Type II} & ~~O(k\log k)
\end{array}
\end{align*}
\end{theorem}

Next, we give necessary and sufficient conditions for a Type II
block to represent an integer $t$. A proof of this result can
also be found in \cite{DH14}.

\begin{lemma}\label{lem:RepresentTTypeII}
Let $\MQ=\begin{pmatrix}2a & b \\ b & 2c \end{pmatrix}$, $b$ odd 
be a type II block, and $t, k$ be positive integers. Then, $\MQ$ represents
$t$ primitively over $\ztkz$ if $\ordt(t)=1$.
\end{lemma}

\section{Symbol of a Quadratic Form}\label{sec:Symbol}

There are several equivalent ways of giving a description of 
the $p^*$-equivalence \cite{CS99, Kitaoka99, OMeara73, Cassels78}. 
In this work,
we go with a modified version of the Conway-Sloane description,
called the $p$-{\em symbol} of a quadratic form. Our modification
gets rid of the need to use the $p$-adic numbers.
Note that $p$-adic numbers are a 
staple in this area and we are not aware of any work which does not
use them \cite{Kitaoka99, OMeara73, Siegel35}.

By definition, two quadratic forms are $p^*$-equivalent if 
they are $p^k$-equivalent for all positive integers $k$. In an
algorithmic sense, this is problematic because 
there are infinitely many possibilities for $k$.
Recall the definition of $k_p$. It equals~1
if $p$ is an odd prime and~3, otherwise. The following theorem
shows that it is enough to test equivalence for just one value of 
$k$. 

\begin{theorem}\label{thm:PStarEq}
Let $\MQ^n$ be an integral quadratic form, $p$ 
be a prime and $k = \ordp(\det(\MQ))+k_p$. If $\MD^n$ is a block 
diagonal form which is equivalent
to $\MQ$ over $\zpkz$, then $\MD \eqp \MQ$.
\end{theorem}
\begin{proof}
Let $\MD^n$ be the block diagonal form equivalent to $\MQ$ over $\zpkz$;
$k=\ordp(\det(\MQ))+k_p$. Then, we show that $\MD$ is $p^\ell$-equivalent
to $\MQ$ for all $\ell > k$.

Let $\MU \in \gln(\zpkz)$ be such that $\MD \equiv \MU'\MQ\MU \bmod{p^k}$.
Then, every entry of $\MU'\MQ\MU-\MD$ must be divisible by $p^k$. Let 
$p^k\tMQ\equiv \MU'\MQ\MU-\MD \bmod{p^\ell}$. Consider the quadratic form
$\MD+p^k\tMQ$. As $k=\ordp(\det(\MQ))+k_p=\ordp(\det(\MD))+k_p$, it follows
that all off-diagonal entries have higher $p$-order than the diagonal 
entries. It is hence possible to diagonalize $\MD+p^k\tMQ$ over 
$\bbZ/p^\ell\bbZ$ to a quadratic form $\MD+p^k\tMD$, where $\tMD$ is also
a block diagonal form with matching Type i.e., if the first block of $\MD$
is Type II then so is the first block of $\tMD$ (see proof of Theorem
\ref{thm:BlockDiagonal}).

Let $\MB, \tilde{\MB}$ be single Type I or Type II blocks. If 
$k=\ordp(d)+k_p$, then $\MB + p^k\tilde{\MB} \eqp \MB$ (see Lemma
\ref{lem:IntegerSymbolIsInvariant} for Type I blocks and Lemma~6, 
\cite{Jones42} for Type II blocks).
Thus, we conclude that
$\MD+p^k\tMD \overset{p^\ell}{\sim} \MD$.
\end{proof}

The rest of this section follows from Conway-Sloane \cite{CS99}
and Theorem \ref{thm:PStarEq}.

Let $\MQ^n$ be an integral quadratic form. The $p$-symbol of $\MQ$
is defined as follows.

\paragraph{$(-1)$-symbol.}
For $p=-1$ the 
$(-1)$-symbol is the same as the $(-1)$-signature
of $\MQ$.

\subsection{$p$-symbol, $p$ odd prime}\label{sec:PSym}

Let $k=\ordp(\det(\MQ))+1$ and $\MD$ be the diagonal quadratic form 
which is $p^k$-equivalent to $\MQ$ (see Theorem \ref{thm:BlockDiagonal}).
Then, $\MD$ can be written as follows.
\begin{align}\label{def:jordanp}
\MD=\MD_0^{n_0} \oplus p\MD_1^{n_1} \cdots \oplus p^i \MD_i^{n_i} \oplus \cdots 
\qquad i \leq \ordp(\det(\MQ))\;,
\end{align}
where $\MD_0,\cdots,\MD_{k-1}$ are diagonal quadratic forms, 
$\sum_i n_i = n$ and $p$ does not divide 
$\det(\MD_0)\cdots\det(\MD_{k-1})$. Let $\scalep(\MQ)$ is the set of
$p$-orders $i$ with non-zero $n_i$.
Then, the $p$-symbol of $\MQ$ is defined as the
set of {\em scales} $i$ occurring in Equation \ref{def:jordanp} with non-zero
$n_i$, {\em dimensions} $n_{i}=\dim(\MD_i)$ and {\em signs} 
$\epsilon_i=\legendre{\det(\MD_i)}{p}$. 
\begin{equation}\label{eq:PSymbol}
\SYM_p(\MQ)= \left\{(p, i,\legendre{\det(\MD_i)}{p},n_i) \mid 
i \in \scalep(\MQ)\right\}
\end{equation}

The following fundamental result follows from Theorem~9, page~379 
\cite{CS99} and Theorem \ref{thm:PStarEq}. 

\begin{theorem}\label{thm:PSymbol}
For $p \in \{-1\} \cup \bbP$, two quadratic forms 
are $p^*$-equivalent iff they have the same $p$-symbol.
\end{theorem}

\subsection{$2$-symbol}\label{sec:2Sym}

Let $k=\ordt(\det(\MQ))+3$ and $\MD$ be the block diagonal form 
which is $2^k$-equivalent to $\MQ$ (Theorem \ref{thm:BlockDiagonal}).
Then, $\MD$ can be written as follows.
\begin{align}\label{def:jordant}
\MD=\MD_0^{n_0} \oplus 2\MD_1^{n_1} \cdots \oplus 2^i \MD_i^{n_i} 
\oplus \cdots \qquad i \leq \ordt(\det(\MQ))\;,
\end{align}
where $\det(\MD_0), \cdots, \det(\MD_i), \cdots$ are odd, $\sum_i n_i=n$
and each $\MD_i$ is in block diagonal form according to Definition 
\ref{def:BlockDiagonal}.
The $2$-symbol of $2^i\MD_i$ are the following quantities.
\begin{equation}\label{def:2invariants}
\left(\begin{array}{lll}
i   & \text{{\em scale} of $\MD_i$} & \\
n_i= \dim(\MD_i) & \text{{\em dimension} of $\MD_i$} &\\
\epsilon_i = \legendre{\det(\MD_i)}{2} & \text{{\em sign} of $\MD_i$} & \\
\type_i=\text{I  or  II} & \text{{\em type} of $\MD_i$} & \text{I, iff 
there is an odd entry on} \\ && \text{~the main diagonal of the }\\
&&\text{~matrix $\MD_i$}\\
\odty_i\in\{0,\cdots,7\} & \text{{\em oddity} of $\MD_i$} & 
\text{it $\type_i=$I, then it 
is equal to} \\ && \text{~the trace of $\MD_i$ read modulo $8$,} \\ 
&& \text{~and is $0$ otherwise.}
\end{array}\right)
\end{equation}
Let the set of scales $i$, with non-zero $n_i$, be denoted 
$\scalet(\MQ)$.
Then, the $2$-symbol of $\MQ$ is written as follows.
\begin{equation}\label{eq:2sym}
\SYM_2(\MQ)= \left\{(2,i,\epsilon_i, n_i, \type_i, \odty_i) 
\mid i \in \scalet(\MQ)\right\}
\end{equation}

In contrast to the $p \in \{-1\} \cup \bbP$ case, 
two $2^*$-equivalent quadratic forms may produce two different 
$2$-symbols. These symbols are then said to be $2$-equivalent.
There is a transformation which maps all equivalent $2$-symbols
to a unique description (see Conway-Sloane \cite{CS99}, page~381). 
We repeat this transformation for the sake of completeness.

\paragraph{Compartments and Trains.} Let $\SymT$ be a $2$-symbol.
Let us define an {\em interval} as a consecutive sequence of 
forms $2^i\MD_i$, even including those with dimension $0$. 
The form with dimension $0$ is treated as a form with
Type II and Legendre symbol $+1$. A 
{\em compartment} is then a maximal interval in which all forms
are of Type I. A {\em train} is a maximal interval with the 
property that for each pair of adjacent forms, at least one is
Type I. 
There are two ways
in which the symbol maybe altered without changing the 
equivalence class.
\begin{enumerate}[(i)]
\item {\bf Oddity fusion.}\, The oddities inside a compartment 
can be changed in such a way that the total sum over any 
compartment remains the same i.e., the sum of oddities in a 
compartment is an invariant. For example, $3\oplus 5$ and $1\oplus 7$
are $2^*$-equivalent.
\item {\bf Sign walking.}\, A $2$-symbol remains in the same equivalence 
class if the signs of any two terms in the same train are simultaneously 
changed, provided certain oddities are changed by~4. Let us suppose that
we want to flip the signs of terms at $2$-scale $i$ and $2$-scale $j$, 
$i<j$. We imagine walking the train from $i$ to $j$, taking steps between
adjacent forms of scales $r$ and $r+1$. Because we are in the same train
during the entire walk, at least one of $\MD_r$ and $\MD_{r+1}$ is of
Type I. The rule is that the total oddity of the compartment must be 
changed by $4$ modulo $8$, each time in the walk when either $\MD_r$
or $\MD_{r+1}$ is in that compartment.

An example from Conway-Sloane is as follows. Here, instead of considering
the $2$-symbol as a tuple $(2, 2^i,\epsilon_i,n_i,\odty_i)$, it is easier to 
consider it as a list with the corresponding term
$(2^i)^{\epsilon_i n_i}_{\odty_i}$. Let us suppose that we
have the following symbol.
\[
1^{+2}_{0}[2^{-2}4^{+3}]_38^{+0}[16^{+1}]_132^{+2}_{0}
\]
The compartments have been denoted by square brackets $[]$ and the symbol 
at scale $3$ has dimension~0. Suppose we want to flip the signs at
scale $1$ and $4$. We have to take the steps $1\to2\to3\to4$. The steps
$1\to 2$ and $2\to3$ uses the first compartment while the step $3\to4$ uses
the second. The oddity of the first compartment remains unchanged (used twice)
and the oddity of the second changes by $4$ modulo $8$. The final equivalent
form is as follows.
\[
1^{+2}_0[2^{+2}4^{+3}]_38^{+0}[16^{-1}]_532^{+2}_0
\]

Using sign walking, one can show that $1\oplus 2^2$ and $5\oplus 2^2\cdot5$ 
are $2^*$-equivalent.
\end{enumerate}

\paragraph{$2$-canonical symbol.}
Using these rules, a $2$-canonical symbol can be computed. This is done
as follows. Compute the $2$-symbol and use oddity fusion and sign walking
to make sure that there is at most one minus sign per train and this is
on the earliest nonzero dimensional form in the train. Using this convention
and only mentioning the total oddities of the compartments, the resulting
description is unique and can be taken as a canonical symbol for the form
(see page~382, \cite{CS99}).
Thus, the $2$-canonical symbol for
$1^{-2}_0[2^{+2}4^{+3}]8^{+0}[16^{+1}]_132^{+2}_0$
is $1^{-2}[2^24^3]_7[16]_132^2$.

\section{Canonicalization: $p$ odd prime}

In this section, we describe the function $\canp$ which
maps an integral quadratic form $\MQ$ to its unique 
canonical form $\canp(\MQ)$. Then, we prove the following
theorem.

\begin{theorem}\label{thm:ALG:CanP}
Let $\MQ^n$ be an integral quadratic form, $p$ be an odd prime and
$k > \ordp(\det(\MQ))$. Then, there is an algorithm (Las Vegas
with constant probability of success)
that given $(\MQ^n, p, k)$ performs
$O(n^{1+\omega}\log k+n\log p+\log^3 p)$ ring operations over $\zpkz$
and outputs $\MU \in \gln(\zpkz)$
such that $\MU'\MQ\MU \equiv \canp(\MQ) \bmod{p^k}$.
\end{theorem}

The $O(\log^3 p)$ ring operations are needed to compute $\sigma_p$, 
assuming GRH. Note that one can fix any quadratic non-residue and
define the canonical form with that instead of $\sigma_p$. This will
obviate the need to use GRH and the number of ring operations
will be $O(n^{1+\omega}\log k+n\log p)$.

The canonical form for an odd prime $p$ is defined as follows.

\begin{definition}\label{def:canp}
A quadratic form is $p$-canonical for an odd prime $p$, if it is of
the form $\oplus_i p^i \MD_i^{n_i}$, where $\MD_i^{n_i}$
is a diagonal quadratic form equal to $\MI^{n_i}$ or 
$\MI^{n_i-1}\oplus \sigma_p$, and the $p$-scales $i$ of the
diagonal entries of the quadratic form are non-decreasing.
\end{definition}

The uniqueness of the $p$-canonical form follows directly from the
definition of $p$-symbol and Theorem \ref{thm:PSymbol}.

We now describe the canonicalization algorithm. Let $\MQ^n$
be an integral quadratic form, $p$ be a prime and $k$ be a positive
integer. For $n=1$, the canonicalization algorithm follows from
Theorem \ref{thm:PolyRep}.

Suppose $n=2$ and
the input $\MQ^2$ is of the form $\tau_1 \oplus \tau_2$, where 
$\tau_1, \tau_2$ are 
units of $\zpz$. The $p$-canonical form of $\MQ$ is 
$1\oplus \{1,\sigma_p\}$.
The following lemma shows how to canonicalize in this case.

\begin{lemma}\label{lem:CanP2Dim}
Let $\tau_1, \tau_2 \in (\zpz)^\times$, $p$ odd prime and $k$ be a positive
integer. Then, there is a $\MU \in \gl_2(\zpkz)$ which transforms
$\tau_1\oplus \tau_2$ to $1\oplus \{1,\sigma_p\}$ modulo $p^k$. The 
transformation $\MU$ can be found by a Las Vegas algorithm which performs
$O(\log k + \log p)$ ring operations over $\zpkz$ and fails with constant
probability.
\end{lemma}
\begin{proof}
Consider the situation when $\tau_1$ is a quadratic residue 
modulo $p$. Then, we use Theorem \ref{thm:PolyRep} to find a primitive
$x$ such that $x^2\tau_1\equiv 1 \bmod{p^k}$. The number $\tau_2$ is either
a non-residue or a residue modulo $p$. In either case, we find a $y$ using
Theorem \ref{thm:PolyRep} again such that 
$y^2\tau_2 \bmod{p^k}\in \{1,\sigma_p\}$. 
\[
\begin{pmatrix}x & 0 \\ 0 & y\end{pmatrix}
\begin{pmatrix}\tau_1 & 0 \\ 0 & \tau_2\end{pmatrix}
\begin{pmatrix}x & 0 \\ 0 & y\end{pmatrix} \equiv
\begin{pmatrix} 1 & 0 \\ 0 & \{1,\sigma_p\}\end{pmatrix} \bmod{p^k}
\]

If $\tau_2$ is a quadratic residue then we make the following 
transformation and reduce to the previous case.
\[
\begin{pmatrix}0 & 1 \\ 1 & 0\end{pmatrix}
\begin{pmatrix}\tau_1 & 0 \\ 0 & \tau_2\end{pmatrix}
\begin{pmatrix}0 & 1 \\ 1 & 0\end{pmatrix} =
\begin{pmatrix}\tau_2 & 0 \\ 0 & \tau_1\end{pmatrix} \bmod{p^k}
\]

Otherwise, both $\tau_1$ and $\tau_2$ are a quadratic non-residues. 
From Lemma \ref{lem:QR},~1 can be written as a sum of two non-residues. 
Let $(\taum_1, \taum_2)$ be one such pair.
Then, we can write~1 as $(\taum_1)+(p^k+1-\taum_1)$ over $\zpkz$,
where both $\taum_1$ and $p^k+1-\taum_1$ are quadratic non-residues
as $\legendre{p^k+1-\taum_1}{p}=\legendre{1-\taum_1}{p}=\legendre{\taum_2}{p}$.
Thus,~1 has a primitive $p^k$-representation in $\tau_1\oplus\tau_2$ over $\zpkz$.

We now use $\tau_1$ to represent $\taum_1$ primitively and $\tau_2$ 
to represent $p^k+1-\tau_1$ primitively over $\zpkz$ (use Theorem 
\ref{thm:PolyRep}).
Let $(x,y) \in (\zpkz)^2$ be the primitive representation. Then, we
extend it to a matrix 
$\begin{pmatrix}x & a \\ y & b\end{pmatrix} \in \gl_2(\zpkz)$, using Lemma
\ref{lem:ExtendPrimitive}. Applying this transformation on 
$\tau_1\oplus \tau_2$ yields the following matrix.
\[
\begin{pmatrix}1 & a\tau_1 x + b \tau_2 y \\
a\tau_1 x + b \tau_2 y & a^2\tau_1+b^2\tau_2\end{pmatrix} \bmod{p^k}
\]
This matrix can be diagonalized using Theorem 
\ref{thm:BlockDiagonal}, keeping the $1$ unchanged; to
a matrix of the following form.
\[
\begin{pmatrix}1 & 0 \\
0 & a\end{pmatrix} \bmod{p^k}
\]
But all these transformation are from $\gl_2(\zpkz)$ and hence 
do not change the symbol of the matrix $\tau_1\oplus \tau_2$. Thus,
$a$ must be a unit of $\zpz$. One can now use Theorem \ref{thm:PolyRep}
to find a $z$ such that $z^2a \bmod{p^k} \in \{1,\sigma_p\}$.
The $\MU$ in this case, is the product of all transformations
in $\gl_2(\zpkz)$ we have used so far.
\end{proof}

We are now ready to prove Theorem \ref{thm:ALG:CanP}.

\begin{proof}(Theorem \ref{thm:ALG:CanP})
The algorithm makes a sequence of transformations, each from 
$\gln(\zpkz)$. 
\begin{enumerate}[(i.)]
\item Use Theorem \ref{thm:BlockDiagonal} to find 
$\MU_0 \in \gln(\zpkz)$ such that $\MU_0'\MQ\MU_0 \bmod{p^k}$
is a diagonal matrix. Use transformations 
$\MV_1, \cdots, \MV_n \in \sln(\zpkz)$ to transform the diagonal
matrix to the form $d_1\oplus \cdots \oplus d_n$,
such that $\ordp(d_1)\leq \cdots \leq \ordp(d_n)$. Note that
each $\MV_i$ exchanges two diagonal entries.
The total number of ring operations for this step is 
$O(n^{1+\omega}\log k)$.
\item The matrix is now of the form 
$p^{i_1}\MD_1^{n_1} \oplus \cdots \oplus p^{i_m}\MD_m^{n_m}$,
where $\MD_1,\cdots,\MD_m$ are diagonal matrices with unit
determinants.
We next use transformations to transform $\MD_i$ to $\canp(\MD_i)$.
This is done as follows.
If $\MD_i$ is of dimension~1 then we use Theorem \ref{thm:PolyRep}
to canonicalize it. 
Otherwise, $\MD_i^{n_i>1}$. The matrix is of the form
$(\tau_1, \cdots, \tau_{n_i})$, where $\tau_1, \cdots, \tau_{n_i}$
are units of $\zpz$. We apply a transformation on 
$\tau_1\oplus \tau_2$ to turn it into $1\oplus \tau$ over $\zpkz$,
where $\tau \in \{1,\sigma_p\}$ (Lemma \ref{lem:CanP2Dim}).
Continuing in a similar way, we end up with the
transformation $\MW_i \in \gl_{n_i}(\zpkz)$, for each $i$ such that 
$\MW_i'\MD_i\MW_i\equiv \MI^{n_i-1}\oplus\{1,\sigma_p\} \bmod{p^k}$.
Let $\MU_i \in \gln(\zpkz)$ be the corresponding local transformation
(see Equation \ref{def:LocalTransformation}). This step takes
$O(n(\log k + \log p))$ ring operations over $\zpkz$.
\end{enumerate}

Then, the transformation $\MU=\MU_0\MV_1 \dots \MV_n \MU_1\cdots$ 
is a product of matrices, each from $\gln(\zpkz)$ and turns $\MQ$ to 
its canonical form. The algorithm performs $O(n^{1+\omega}\log k + n\log p)$
ring operations over $\zpkz$. It additionally needs $O(\log^3 p)$ ring 
operations to compute $\sigma_p$, assuming GRH.
\end{proof}

\section{Canonicalization: $p=2$}

The $2$-canonical form is non-trivial and requires care. We follow
the same procedure (sign walking and oddity fusion) as 
Conway-Sloane \cite{CS99} given in Section
\ref{sec:2Sym}. Then, we prove the following theorem.

\begin{theorem}\label{thm:ALG:Can2}
Let $\MQ^n$ be an integral quadratic form, and
$k \geq \ordt(\det(\MQ))+3$. Then, there is an algorithm 
that given $(\MQ^n, k)$ performs
$O(n^{1+\omega}\log k+nk^3)$ ring operations over $\ztkz$
and outputs $\MU \in \gln(\ztkz)$
such that $\MU'\MQ\MU \equiv \cant(\MQ) \bmod{2^k}$.
\end{theorem}

\subsection{Type II Block}

Our next step is to give a transformation which maps any Type II matrix
to its canonical form. Recall that the canonical form of a Type II matrix
$\MQ$ is either $2^{\ordt(\MQ)}\MTM$ or $2^{\ordt(\MQ)}\MTP$.

\begin{lemma}\label{lem:TypeIICanonical}
Let $\MQ$ be a Type II matrix of $2$-order~0. Then, 
there is an algorithm that given $(\MQ, k\geq 3)$ as
input; performs $O(k\log k)$ ring operations and outputs a 
$\MU \in \gl_2(\ztkz)$ such that 
$\MU'\MQ\MU \bmod{2^k} \in \{\MTP, \MTM\}$.
\end{lemma}
\begin{proof}
By definition, $\MQ$ is of the form
$\begin{pmatrix}2a & b \\ b & 2c\end{pmatrix}$, where $b$ is odd. 
The integer $b$ is odd and so $b^2 \bmod{8}=1$. Thus,
$\det(\MQ) \bmod 4=4ac-b^2 \bmod4=3$ and $\det(\MQ) \bmod8 \in \{3,7\}$. 
For convenience, suppose that
\[
\lambda = \det(\MQ) \bmod8, \qquad \lambda \in \{3,7\}\;.
\]


Let $s=k+1$. We now give a transformation which maps $\MQ$ to its
canonical form. 
\begin{enumerate}[(i.)]
\item From Lemma \ref{lem:RepresentTTypeII}, $2$ has a 
primitive representation in $\MQ$ over $\ztsz$. Use Theorem
\ref{thm:PolyRep} to find one such primitive 
representation $(x_1,x_2) \in (\ztsz)^2$. Without loss of 
generality assume $x_1$ is odd and define $\MU \in \glt(\ztsz)$
as follows.
\[
\MU = \begin{pmatrix} x_1 & 0 \\ x_2 & x_1^{-1} \bmod{2^s} \end{pmatrix},
\MU'\MQ\MU \equiv \begin{pmatrix}2 & b+2cx_2x_1^{-1}\\
b+2cx_2x_1^{-1} & 2cx_1^{-2} \end{pmatrix} \bmod{2^s}
\]
\item The matrix $\MU$ is in $\glt(\ztsz)$. Thus, $\det(\MU)$ is a
unit of $\ztsz$ and $\det(\MU'\MQ\MU) \bmod8=
\lambda=\det(\MQ) \bmod8$. Thus, the following equation has a 
solution.
\begin{align}\label{EQ:TypeIICanonical:1}
x^2 \det(\MU'\MQ\MU) \equiv \lambda \pmod{2^s}
\end{align}
\item A primitive solution of Equation \ref{EQ:TypeIICanonical:1} can be found
using Theorem \ref{thm:PolyRep}. Let us denote the solution by $x$.
Then, $x$ is primitive and the matrix $\MV$ defined by $1\oplus x$ is
in $\glt(\ztsz)$. 
\item Let $\MS := \MV'\MU'\MQ\MU\MV \bmod{2^s}$. Then,
by construction, $\MS_{11}=2$ and 
$\det(\MS)=x^2\det(\MU'\MQ\MU) \bmod{2^s}$ which equals $\lambda$ 
(Equation \ref{EQ:TypeIICanonical:1}).
\item By assumption, $b$ is odd. But then, $\MS_{12}$ is odd and 
$(1-\MS_{12})/2$ is an element of $\ztsz$. It follows that the matrix 
$\MW:=\begin{pmatrix}1 & \frac{1-\MS_{12}}{2} \\ 0 & 1\end{pmatrix}$ 
is in $\glt(\ztsz)$. But then, for some integer $y$,
\begin{align}\label{EQ:TypeIICanonical:2}
\MW'\MS\MW \equiv \begin{pmatrix}2 & 1 \\ 1 & 2y\end{pmatrix} \bmod{2^s}
\end{align}
\item By construction, $\det(\MW)\equiv 1 \bmod{2^s}$ and from item (iv),
\begin{align}\label{EQ:TypeIICanonical:3}
\det(\MW'\MS\MW) \equiv \det(\MS) \equiv \lambda \pmod{2^s}
\end{align}
\item By Equation \ref{EQ:TypeIICanonical:2} and Equation 
\ref{EQ:TypeIICanonical:3}, $4y-1\equiv \lambda \pmod{2^s}$. Recall 
$\lambda \in \{3,7\}$.
This implies that $\lambda+1$ is divisible
by~4 and 
\[
2y \bmod{2^s} 
\in \left\{\begin{array}{ll}
\{2, 2+2^{s-1}\} & \lambda=3 \\
\{4, 4+2^{s-1}\} & \lambda=7 
\end{array}\right.
\implies
2y \bmod{2^k} 
= \left\{\begin{array}{ll}
2 & \lambda=3 \\
4 & \lambda=7 
\end{array}\right.
\]
\item Thus, the transformation $\MU\MV\MW$ is in $\glt(\ztkz)$ and
transforms $\MQ$ to its canonical form over $\ztkz$, by construction.
\end{enumerate}

By Lemma \ref{lem:RepresentTTypeII} and Theorem \ref{thm:PolyRep},
the transformation can be constructed in $O(k\log k)$ ring operations.
\end{proof}

\subsection{Dimension$=3$, with one Type II block}

Let us suppose that the input matrix $\MQ^3$ is of the form
$\tau \oplus \MTT$, where $\tau$ is a unit of $\ztkz$. In this case,
we show that the matrix can be transformed into a diagonal matrix
$\tau_1 \oplus \tau_2 \oplus \tau_3$ over $\ztkz$ such that
$\tau_1,\tau_2,\tau_3$ are all units of $\ztkz$.

\begin{lemma}\label{lem:Can2dim3}
Let $k \geq 3$ be an integer and $\tau$ be a unit of $\ztkz$. Then,
there is an algorithm that performs $O(k\log k)$ ring operations and
transforms $\tau \oplus \MTT$ to $\tau_1\oplus\tau_2\oplus\tau_3$, 
where $\tau_1,\tau_2,\tau_3$ are units of $\ztkz$.
\end{lemma}
\begin{proof}
As usual, we provide a sequence of transformations, each from
$\gl_3(\ztkz)$. 
\begin{enumerate}[(i.)]
\item Let $\MV \in \glt(\ztkz)$ be the matrix that transforms
$\begin{pmatrix}8 & 1 \\ 1 & 2\end{pmatrix}$ to $\MTP$ over
$\ztkz$. This matrix can be found using Lemma 
\ref{lem:TypeIICanonical}. Suppose $\MU^3$ is defined as follows.
\[
\MU = \left\{\begin{array}{ll}
1\oplus \MU_1^{-1} \bmod{2^k} & \text{if }\MTT=\MTP\\
\MI^3 & \text{otherwise}
\end{array}\right.
\]
By construction, $\MU$ transforms $\tau \oplus \MTT$ to the
following form.
\[
\tau\oplus\MTT \underset{\MU,2^k}{\to}\begin{pmatrix}\tau & 0 & 0 \\
0 & x & 1 \\ 0&1&2\end{pmatrix},\qquad \text{where }x \in \{2,8\}
\]
\item Consider the matrix $\MV$ defined as follows.
\[
\MV = \begin{pmatrix}1 & 1 & 1 \\ r & 1 & 0\\ 0 & 1 & 1\end{pmatrix}
\]
Whatever integer value $r$ takes, the matrix $\MV$ has determinant~1
and hence $\MV \in \gl_3(\ztkz)$. The matrix $\MV$ makes the
following transformation.
\[
\begin{pmatrix}\tau & 0 & 0 \\0 & x & 1 \\ 0&1&2\end{pmatrix} 
\underset{\MV,2^k}{\to} 
\begin{pmatrix}r^2x+\tau & rx+r+\tau & \tau+r \\ 
rx+r+\tau & \tau+4+x & \tau+3 \\ 
\tau+r	&	\tau+3	& \tau+2\end{pmatrix}
\]
\item The integer $x \in \{2,8\}$ and so $(x+1)$ is a unit of $\ztkz$. 
If we set $r:=\frac{-\tau}{x+1} \bmod{2^k}$, then 
$rx+r+\tau \equiv 0 \bmod{2^k}$. Thus, the matrix has been transformed
into the following form, where $\tau_1=r^2x+\tau$ and $\tau_2=\tau+4+x$.
\[
\begin{pmatrix}\tau_1 & 0 & \tau+r \\ 
0 & \tau_2 & \tau+3 \\ 
\tau+r	&	\tau+3	& \tau+2\end{pmatrix}
\]
\item The numbers $\tau_1=r^2x+\tau$ as well as $\tau_2=\tau+4+x$ are
odd because $x \in \{2,8\}$ and $\tau$ is odd. Thus, by Theorem
\ref{thm:BlockDiagonal}, we can find a matrix $\MW \in \gl_3(\ztkz)$
such that $\MW$ transforms the matrix to the following final form,
where $\tau_3$ is also a unit of $\ztkz$.
\[
\begin{pmatrix}\tau_1 & 0 & 0 \\ 
0 & \tau_2 & 0 \\ 
0	&	0	& \tau_3\end{pmatrix}
\]
\end{enumerate}
\end{proof}

Lemma \ref{lem:Can2dim3} implies that one does not need to have
Type I and Type II matrices on the same $2$-scale. Thus, for every
$2$-symbol there is an equivalent $2$-symbol in which all $2$-scales
are either exclusively Type I or Type II.

\subsection{Sign Walking}\label{sec:SignWalking}

We now make the transformations required to perform sign 
walking. Note that the sign walking is done on
symbols but we make the corresponding transformations between
$2^k$-equivalent quadratic forms. Also,
sign walking involves a
single train and hence, we never walk between two Type II forms. 

By definition, $\legendre{\tau}{2} = 1$ iff 
$\tau \bmod 8 \in \{1,7\}$. Thus, $\tau+4$ always has the 
opposite sign of $\tau$. In the lemma below, notice that
in every transformation the negative sign propagates to
the front.


\begin{lemma}\label{lem:SignWalk}
There exists invertible transformations over $\ztkz$, computable
in $O(k\log k)$ ring operations, such that
\begin{align*}
\begin{array}{lllll}
(i.) & \tau_1\oplus 4\tau_2 &\totk  &(\tau_1+4) \oplus 4(\tau_2+4) & k \geq 5\\
(ii.) & \tau\oplus 2\MTM &\totk &(\tau+4)\oplus 2\MTP & k \geq 4\\
(iii.) & \MTT_1\oplus 2\taum_1 &\totk &\MTT_2 \oplus 2\taup_2 & k \geq 4\\
(iv.) & \MT_1 \oplus \MTT & \totk & \MT_2 \oplus \MTP & k \geq 3
\end{array}
\end{align*}
where $\tau_1,\tau_2\in \{1,3,5,7\}$, $\tau_1^{-}\in\{3,5\}$, 
$\tau_2^{+} \in \{1,7\}$ and
$\MTT_1, \MTT_2 \in \{\MTM, \MTP\}$.
\end{lemma}
\begin{proof}
We itemize the transformation and show how to find them under the
corresponding item. 
\begin{enumerate}[(i.)]
\item By assumption, $\tau_1$ and $\tau_2$ are odd and so $\tau_2$ is
invertible over $\ztkz$. Let $\MU$ be defined as follows.
\begin{align}\label{EQ:SWi:U}
\MU:=\begin{pmatrix}1 & 4 \\ 1 & -\frac{\tau_1}{\tau_2} \bmod{2^{k-2}}\end{pmatrix}
\bmod{2^k}
\end{align}
Then, $\det(\MU)\equiv-4-\frac{\tau_1}{\tau_2} \bmod{2^{k-2}}$;
$\det(\MU)$ is odd and $\MU \in \glt(\ztkz)$.
If $x=-\frac{\tau_1}{\tau_2} \bmod{2^{k-2}}$, then
$4(\tau_2x+\tau_1) \equiv 0 \bmod{2^k}$. The transformation
$\MU$ has the following effect on our input matrix $\tau_1+4\tau_2$.
\[
\begin{pmatrix}\tau_1 & 0 \\ 0 & 4\tau_2\end{pmatrix}
\underset{\MU,2^k}{\to}
\begin{pmatrix}\tau_1+4\tau_2 & 0\\
0 & 4(\tau_2x^2+4\tau_1)\end{pmatrix}
\]
The integers $\tau_1, \tau_2$ are odd by hypothesis and $x$ is odd by 
construction. But then, 
\begin{align*}
\tau_1+4\tau_2 & \equiv \tau_1+4 \bmod 8 \\
\tau_2 x^2 + 4\tau_1 & \equiv \tau_2 + 4 \bmod 8
\end{align*}
Thus, we can find $2^k$-primitive transformations that map $\tau_1+4\tau_2$
to $\tau_1+4$ and $4(\tau_2x^2+4\tau_1)$ to $4(\tau_2+4)$ using Lemma
\ref{thm:PolyRep}. These transformations take the matrix to
the final form.

\item We first transform our input matrix $\tau\oplus\MTM$ by the 
matrix $\MV \in \glt(\zpkz)$ defined below.
\[
\MV=\begin{pmatrix}1 & 0 & 0 \\ 1 & 1 & 0 \\ 0 & 0 & 1\end{pmatrix}
\qquad
\MV'\begin{pmatrix}\tau\oplus \MTM\end{pmatrix}\MV=
\begin{pmatrix}\tau+4 & 4 & 2 \\ 4 & 4 & 2\\ 2 & 2 & 4\end{pmatrix}
\]
This transformation brings the first entry to the correct sign. 
We block diagonalizing this matrix using our algorithm in the proof of Theorem
\ref{thm:BlockDiagonal} and find a matrix $\MU\in \gl_3(\ztkz)$ such
that
\[
\MU'
\begin{pmatrix}\tau+4 & 4 & 2 \\ 4 & 4 & 2\\ 2 & 2 & 4\end{pmatrix}
\MU \equiv 
(\tau+4) \oplus \MX \pmod{2^k} \;,
\]
where $\MX$ is a Type II block.
The transformations made so far are from $\gl_3(\ztkz)$, and so,
$\det(\tau_1\oplus2\MTM) \overset{2^k}{\sim} (\tau+4)\det(\MX)$.
In particular, this implies that $\ordt(\det(\MX))=2$ and 
$\legendre{\copt(\det(\MX))}{2}=+$. Because, $\MX$ is a Type II
block, it is $2^*$-equivalent to $2\MTP$ and such a transformation
can be found using Lemma \ref{lem:TypeIICanonical}.

\item Given the input matrix $\MTT_1 \oplus 2\tau_1^{-}$ we
apply the transformation $\MV \in \gl_3(\ztkz)$ given below.
\begin{gather*}
\MV=\begin{pmatrix}1 & 0 & 0 \\ 0 & 1&0\\0&1&1\end{pmatrix}, 
\begin{pmatrix}2 & 1 & 0 \\ 1 & b \in \{2,4\}&0\\0&0&2\tau\end{pmatrix}
\underset{\MV,2^k}{\to}
\begin{pmatrix}2 & 1 & 0\\1 & b+2\tau & 2\tau\\0 & 2\tau & 2\tau\end{pmatrix}
\end{gather*}
This transformation maps $b$ to $b+2\tau$, where $\tau$ is odd. But then,
\[
\det(\MTT_1) = 2b-1 \qquad 
\det\begin{pmatrix}2 & 1 \\ 1 & b+2\tau\end{pmatrix} = 2b-1+4\tau
\]
Thus, the sign of the Type II matrix has been switched. Next we use the
block diagonalization algorithm from the proof of Theorem 
\ref{thm:BlockDiagonal}. If $\MU$ is the output then, for some integer
$x$,
\[
\MU'\begin{pmatrix}2 & 1 & 0\\1 & b+2\tau & 2\tau\\0 & 2\tau & 2\tau
\end{pmatrix}\MU \equiv 
\begin{pmatrix}2 & 1 & 0 \\ 1 & b+2\tau & 0 \\ 0 & 0 & x\end{pmatrix}
\pmod{2^k}
\]
As before, the $2^k$-symbol of the determinant of the quadratic 
form does not change when transformed by $\MU \in \gl_3(\ztkz)$;
implying, $\ordt(x)=1$ and $\sgnt(\copt(x))=+$. Using Lemma
\ref{lem:TypeIICanonical} and Theorem \ref{thm:PolyRep},
the following transformation over $\ztkz$ can be found.
\[
\begin{pmatrix}2 & 1 & 0 \\ 1 & b+2\tau & 0 \\ 0 & 0 & x\end{pmatrix}
\to 
\MTT_2 \oplus 2\tau_2^+
\]
where $\MTT_2$ has the opposite sign as $\MTT_1$ and 
$\tau_2^+ \in \{1,7\}$.

\item Let $b \in \{2,4\}$, then the input matrix $\MTT_1\oplus\MTM$ is in the 
following form.
\[
\begin{pmatrix}2 & 1 \\ 1 & b\end{pmatrix}\oplus
\begin{pmatrix}2 & 1 \\ 1 & 2\end{pmatrix}
\]
If $b=4$ then we swap the two matrices. Otherwise, we apply the
following transformation.
\begin{gather*}
\MV = \begin{pmatrix}1 & 0 & 0 & 0\\ 0&1&0&0\\0&1&1&0\\0&0&0&1\end{pmatrix},~~
\MTM\oplus \MTM\underset{\MV,2^k}{\to}
\begin{pmatrix}2&1&0&0\\1&4&2&1\\0&2&2&1\\0&1&1&2\end{pmatrix}
\end{gather*}
Diagonalization (see Theorem \ref{thm:BlockDiagonal}) of this matrix 
from top down, will yield a quadratic
form of the form $\MTP \oplus \MY$, where $\MY$ is also equivalent
to $\MTP$. We then make a local transformation to convert $\MY$ to
$\MTP$.
\end{enumerate}
\end{proof}

\subsection{Oddity Fusion}\label{sec:OddityFusion}

The transformations under the oddity fusion step deal with 
a single compartment.
A compartment is a consecutive sequence
of Type I forms. Two adjacent quadratic forms in the same
compartment differ by at most~1 in terms of their $2$-scale.
In this case, we want to find the minimum lexicographically
possible set of integers, that can be represented.

\begin{lemma}\label{lem:TypeIIdim3}
Let $\tau, \tau_1, \tau_2,\tau_3 \in \SGNI$, $i, i_1, i_2, i_3$ be 
positive integers and $\MTT \in \{\MTM, \MTP\}$ be a Type II
matrix. If $2^{i_1}\tau_1\oplus2^{i_2}\tau_2\oplus2^{i_3}\tau_3 \eqt
2^{i_1}\tau\oplus2^i\MTT$, then $i_1=i_2=i_3=i=1$.
\end{lemma}
\begin{proof}
Suppose that $i_1=i_2=i_2=i$ is not true. Then, the $2$-symbols 
(see Section \ref{sec:2Sym}) of the first and
the second quadratic form are not equivalent because
one cannot be transformed into the other by a combination of
oddity fusion and sign walking steps.
\end{proof}

\begin{lemma}\label{lem:CanTypeIDim32}
Let $\tau_1,\tau_2,\tau_3$ be odd integers and $k$ be a positive 
integer. Then, there is an algorithm that transforms
$\MD=\tau_1\oplus\tau_2\oplus\tau_3$ to one of the forms
in Table \ref{tab:Cdim3}
in $O(k^2\log k)$ ring operations, where 
$\epsilon=\legendre{\tau_1\tau_2\tau_3}{2}$ and 
$\odty:=\tau_1+\tau_2+\tau_3 \bmod 8$.
\begin{table}[h]
\caption{Type I Canonical Forms for $n=3$}
\centering
\begin{tabular}{| l | l | l || l|l|l|}
\hline
$\epsilon$ & $\odty$ & Form & $\epsilon$ & $\odty$ & Form\\
\hline
$+$ & 1 & $1\oplus1\oplus7$ & $-$ & 1 &$3\oplus3\oplus3$\\
$+$ & 3 & $1\oplus1\oplus1$ & $-$ & 3 & $1\oplus3\oplus7$\\
$+$ & 5 & $3\oplus3\oplus7$ & $-$ & 5 & $1\oplus1\oplus3$\\
$+$ & 7 & $1\oplus3\oplus3$ & $-$ & 7 & $1\oplus1\oplus5$\\
\hline
\end{tabular}
\label{tab:Cdim3}
\end{table}
\end{lemma}
\begin{proof}
The forms listed in Table \ref{tab:Cdim3} are exhaustive. The
transformation from $\MD$ to one
of these forms can be done using Theorem \ref{thm:PolyRep}
as follows.
\begin{enumerate}[(i.)]
\item Read the canonical form from Table \ref{tab:Cdim3} using
the oddity and the value of $\epsilon$ of the quadratic form $\MD$.
Let it be $\ttau_1\oplus\ttau_2\oplus\ttau_3$.
\item Use Theorem \ref{thm:PolyRep} to represent $\ttau_1$ using
$\MD$ over $\ztkz$. Let $\Vx \in (\ztkz)^3$ be the representation
and $\MU \in \gl_3(\ztkz)$ be the corresponding primitive 
extension as in Lemma \ref{lem:ExtendPrimitive}. 
\item The integer $\ttau_1$ is odd and hence using Theorem 
\ref{thm:BlockDiagonal} the matrix $\MU'\MD\MU$ can be block 
diagonalized over $\ztkz$ by matrix $\MV \in \gl_3(\ztkz)$ 
such that;
\[
\MV'\MU'\MD\MU\MV \equiv \ttau_1 \oplus \MB
\text{, where } \MB \in (\ztkz)^{2\times 2}
\]
\item The oddity and Legendre symbol for the matrix $\MB$ can
be computed exhaustively using $\MD$ and $\ttau_1$ as follows.
\begin{align*}
\odty(\MB) &= \odty(\MD) - \ttau_1 \bmod 8\\
\legendre{\det(\MB)}{2} &= \epsilon\legendre{\ttau_1}{2}
\end{align*}
\item If $\odty(\MB)=0$ then $\MB$ might be of Type II. The 
exhaustive list of such matrices $\ttau_1\oplus \MB$, where
$\MB$ is Type II is given in Table \ref{tab:CBadTypeII3}, below.
\begin{table}[h]
\caption{Bad cases for $n=3$}
\centering
\begin{tabular}{| l | l | l | l |}
\hline
Form with Type II & Equivalent Form & $\odty$ & $\epsilon$ \\ 
\hline
\hline
$1\oplus\MTM$ &  $3\oplus3\oplus3$ & 1 & $-$ \\
$3\oplus\MTM$ &  $1\oplus1\oplus1$ & 3 & $+$ \\
$5\oplus\MTM$ &  $3\oplus3\oplus7$ & 5 & $+$ \\
$7\oplus\MTM$ &  $1\oplus1\oplus5$ & 7 & $-$ \\
\hline
$1\oplus\MTP$ &  $1\oplus1\oplus7$ & 1 & $+$ \\
$3\oplus\MTP$ &  $1\oplus3\oplus7$ & 3 & $-$ \\
$5\oplus\MTP$ &  $1\oplus1\oplus3$ & 5 & $-$ \\
$7\oplus\MTP$ &  $1\oplus3\oplus3$ & 7 & $+$ \\
\hline
\end{tabular}
\label{tab:CBadTypeII3}
\end{table}
\item The bad cases are problematic because it is impossible to
transform $\MTM$ or $\MTP$ to a form $\ttau_2\oplus\ttau_3$ using
transformations from $\gl_2(\ztkz)$. Fortunately, the strategy
to represent the smallest possible $\ttau_1$ fails i.e., results in
one of the bad cases; only when $\MD \eqt 1\oplus1\oplus7$. For all
other forms in Table \ref{tab:Cdim3} it can be checked that 
$\ttau_2+\ttau_3 \bmod 8 \neq 0$. 

For the special case of $1\oplus1\oplus7$, we represent
~7 instead of~1 under item (ii.). Then, $\odty(\MB)=1+1=2$
and $\MB$ is not of Type II.

\item If $\MB$ is not of Type II then we transform $\MB$ to
$\ttau_2 \oplus \ttau_3$ using Theorem \ref{thm:PolyRep} and
Theorem \ref{thm:BlockDiagonal}. 
Note that the transformation
exists because the $2$-symbol of $\MB$ matches the $2$-symbol of
$\ttau_2 \oplus \ttau_3$.
In case of $\MD=1\oplus1\oplus7$
we end up with $7\oplus1\oplus1$ instead. We then swap $7$ and $1$
using a transformation from $\gl_3(\ztkz)$. 
\end{enumerate}
\end{proof}

\subsection{Canonicalizing a Single Compartment}

By definition, all forms in a compartment are of scaled Type I i.e.,
the compartment is of the form 
$2^{i_1}\tau_1\oplus 2^{i_2}\tau_2 \oplus \cdots$,
where $\tau_1,\tau_2, \cdots \in \SGNI$ and $i_1,i_2,\cdots$ are positive 
integers.

\begin{definition}\label{def:CanCompartment}
Let $\MD=2^{i_1}\tau_1\oplus 2^{i_2}\tau_2 \oplus \cdots \oplus 
2^{i_n}\tau_n$ be a single
compartment, where $\tau_1,\cdots, \tau_n\in \SGNI$ and 
$i_1\leq \cdots\leq i_n$ are positive integers such that any two consecutive
ones differ by at most~1. Then, the canonical form
of $\MD$ is 
$2^{i_1}\ttau_1\oplus \cdots \oplus 2^{i_n}\ttau_n$, where
$(\ttau_1,\cdots, \ttau_n)$ is lexicographically minimum possible
option in the $2^*$-equivalence class of $\MD$.
\end{definition}

\begin{lemma}\label{lem:CanDim4}
Let $k\geq 3$ be an integer, $\tau \in \SGNI$ and 
$\MD^n=\tau_1\oplus2^{i_2}\tau_2\oplus\cdots\oplus2^{i_n}\tau_n$ be a 
diagonal form with $\tau_1,\cdots,\tau_n \in \SGNI$,
and $i_2\leq \cdots \leq i_n$. 
Then, $\tau$ is primitively 
representable in $\MD$ over $\ztkz$, if it is primitively representable
in $\tau_1\oplus\cdots\oplus2^{i_4}\tau_4$ over $\ztkz$. 
\end{lemma}
\begin{proof}
A primitive representation of $\tau$ exists iff a primitive 
representation exists modulo $8$ (Theorem \ref{thm:Siegel13}).
Let $\MA=\tau_1\oplus\cdots\oplus2^{i_4}\tau_4$.
We can make two simplifications to the form $\MD$: 
(i) we can only use Type I blocks for which $2$-order is
$\leq 2$, and (ii) there is no need to have more than~1
element of order~2 as we can add at most~4 modulo~8 to
the result using any number of such elements. Item (ii)
implies that if $2 \in \{i_2,i_3,i_4\}$ then we do not
need to use $2^{i_5}\tau_5, \cdots, 2^{i_n}\tau_n$ i.e.,
the lemma is true in this case. 

Thus, (i)+(ii) imply that the only
possible values for the vector $(i_2,i_3,i_4)$ are
$(0,0,0), (0,0,1), (0,1,1), $ and $(1,1,1)$. 
In all these cases, and for all possible values of 
$\tau_1, \cdots, \tau_4 \in \SGNI$, we verify by brute-force
that~1 can be represented modulo~8. Let $\Vx \in (\ztkz)^4$
be a primitive $2^k$-representation of $1$ and 
$\MU \in \gl_4(\ztkz)$ be an 
extension of $\Vx$ given by Lemma \ref{lem:ExtendPrimitive}.
By construction, $(\MU'\MA\MU)_{11} \equiv 1 \bmod{2^k}$. Using
Theorem \ref{thm:BlockDiagonal} we can find $\MV \in \gl_4(\ztkz)$
such that 
\[
\MV'\MU'\MA\MU\MV \equiv 1\oplus\MX \bmod{2^k},
\text{ where } \MX \text{ is in block diagonal form}.
\]

If $\MX$ does not have a Type II block then we are done.
Depending on the value of $(i_2,i_3,i_4)$ we proceed as
follows.
\begin{description}
\item[$(0,0,0)$.] If a Type II form appears within $\MX$
then $\MX=\MT\oplus\tau$, where $\MT$ is a Type II block of~2-order~0.
In this case, we can locally get rid of the Type II block
using Lemma \ref{lem:Can2dim3} as follows.
\begin{align}\label{EQ:Can2dim3}
1\oplus \MX \underset{2^k}{\to} 1\oplus a \oplus b \oplus c
\end{align}
\item[$(1,1,1)$.] If a Type II form appears within $\MX$
then $\MX=2\MT\oplus2\tau$, where $\MT$ is a Type II block of~2-order~0.
We locally get rid of the Type II block using Lemma \ref{lem:Can2dim3}
as in Equation \ref{EQ:Can2dim3}.
\item[$(0,1,1)$.] It is impossible for $\MX$ to contain a Type II
form because then the $2$-symbols of $\MA$ and $1\oplus\MX$ will
not be equivalent i.e., they cannot be transformed using sign walking
and oddity fusion.
\item[$(0,1,1)$.] The matrix $\MA$ is of the form 
$\tau_1\oplus\tau_2\oplus\tau_3\oplus2\tau_4$,
in this case.
If $\tau_1+\tau_2+\tau_3 \bmod 8 \not\in \{3,7\}$ then we apply
the following transformation.
\begin{gather*}
\MW = \MI^{2\times2}\oplus\begin{pmatrix}1&0\\1&1\end{pmatrix},~~
(\tau_1\oplus\tau_2\oplus\tau_3\oplus2\tau_4) \underset{\MW,2^k}{\to}
\tau_1\oplus\tau_2\oplus\begin{pmatrix}\tau_3+2\tau_4&2\tau_4\\2\tau_4&2\tau_4\end{pmatrix}\\
\overset{\text{Theorem } \ref{thm:BlockDiagonal}}{\underset{2^k}{\to}}
\tau_1\oplus\tau_2\oplus(\tau_3+2\tau_4)\oplus2\tau_5
\end{gather*}
But then, $\tau_1+\tau_2+\tau_3+2\tau_4\equiv \tau_1+\tau_2+\tau_3+2\bmod{4}$.
Thus, we may assume that the sum of the first three Type I 
entries of the input quadratic form $\MA$ is in the set $\{3,7\}$ 
modulo~8 i.e., if $\MA=\tau_1\oplus\tau_2\oplus\tau_3\oplus2\tau_4$,
then $\tau_1+\tau_2+\tau_3 \bmod 8 \in \{3,7\}$. In this case, we 
exhaustively check that~1 can be represented primitively
using only $\tau_1\oplus\tau_2\oplus\tau_3$. It also follows that the
oddity of the leftover $2\times 2$ matrix must be $2$ or $6$. But then, this
matrix cannot be Type II.
\end{description}
\end{proof}

\begin{lemma}\label{lem:CanCompartment}
Let $\MD=2^{i_1}\tau_1\oplus \cdots \oplus 2^{i_n}\tau_n$ be a single
compartment, where $\tau_1,\cdots, \tau_n\in \SGNI$ and 
$i_1\leq \cdots\leq i_n, k$ are positive integers. Then, there is
an algorithm that performs $O(nk^3)$ ring operations and
finds $\MU \in \gl_n(\ztkz)$ that 
transforms $\MD$ into $\cant(\MD)$ over $\ztkz$.
\end{lemma}
\begin{proof}
We divide the proof in several cases, depending on the value of the
dimension $n$.
\begin{description}
\item[$n=2$.] We exhaustively try to primitively represent the smallest 
integer of the form $2^{i_1}\tau$, where $\tau \in \SGNI$ using Theorem
\ref{thm:PolyRep}. Let $\Vx \in (\ztkz)^2$ be a representation, 
$\MU \in \gl_2(\ztkz)$ be the corresponding extension and $\MV$ be the
transformation given by the block diagonalization Theorem 
\ref{thm:BlockDiagonal}, then
\[
\MV'\MU'\MA\MU\MV \equiv 2^{i_1}\tau \oplus 2^{i_2}\ttau \bmod{2^k},
\text{ where }\ttau \text{ is odd}\;.
\]
We now use Theorem \ref{thm:PolyRep} to transform $\ttau$ locally
to something from the set $\SGN^{\times}$ over $\ztkz$. By construction,
the resulting matrix is in its $2$-canonical form.

\item[$n=3$.]
The sequence of transformations is as follows: (i) find the smallest
primitively $2^k$-representable integer of the form $2^{i_1}\ttau_1$
with $\ttau_1 \in \SGNI$ by doing an exhaustive search for primitive
representation of $\ttau_1$ by $\tau_1\oplus 2^{i_2-i_1}\tau_2 \oplus
2^{i_3-i_1}\ttau_3$ over $\bbZ/8\bbZ$ (Theorem \ref{thm:Siegel13}),
(ii) Find a primitive representation $\Vx \in (\ztkz)^3$ using Theorem
\ref{thm:PolyRep} and extend it to $\MU \in \gl_3(\ztkz)$ using Lemma
\ref{lem:ExtendPrimitive}, (iii) Block diagonalize $\MU'\MD\MU$
using $\MV$ given by Theorem \ref{thm:BlockDiagonal}. Then,
\[
\MV'\MU'\MD\MU\MV \equiv 2^{i_1}\ttau_1 \oplus \MX, \text{ where }
\MX \in (\ztkz)^{2\times2}
\]
The type of $\MX$ is II only when $i_1=i_2=i_3$ (Lemma \ref{lem:TypeIIdim3}).
If this is the case, then we can apply Lemma \ref{lem:CanTypeIDim32} 
to canonicalize instead. Otherwise, $\MX$ is of Type I and we have
reduced to the case of $n=2$.

\item[$n\geq 4$.] By Lemma \ref{lem:CanDim4}, we can represent the smallest
possible integer of the form $2^{i_1}\tau$, with $\tau \in \SGNI$. This way
we reduce to one smaller dimension. Finally, we reduce to the case of
dimension~3.
\end{description}
The number of ring operations follows from using Theorem \ref{thm:PolyRep} 
at most $O(n)$ times on diagonal matrices of dimensions at most~4.
\end{proof}

\subsection{Canonical Form, any dimension} 

We can now define the function $\cant(\MQ^n)$. The uniqueness
follows from Conway-Sloane \cite{CS99}, see Section \ref{sec:2Sym}.

\begin{proof}(Theorem \ref{thm:ALG:Can2})
We perform the following sequence of transformations over $\ztkz$.
\begin{enumerate}
\addtolength{\itemsep}{-4pt}
\item Block diagonalize the quadratic form.
\item For each type II block, apply the transform it to $\MTP$ or $\MTM$ 
using Lemma \ref{lem:TypeIICanonical}.
\item For each $2$-scale, apply the transformation in Lemma
\ref{lem:Can2dim3} to transform the matrix to a block diagonal form
where all scales have either only type I matrices, or only type II matrices.

\item For each train, do a sign walk to move all minus signs to the front
of the train (see Lemma \ref{lem:SignWalk}).
Also, from Lemma \ref{lem:SignWalk}, the canonical form for each type 
II part has at most one $\MTM$ i.e., it is either
$2^i(\MTM,\MTP,\cdots,\MTP)$ or $2^i(\MTP,\MTP, \cdots, \MTP)$.

\item Transform each compartment to its corresponding canonical form
(Definition \ref{def:CanCompartment}) using Lemma \ref{lem:CanCompartment}.
\addtolength{\itemsep}{4pt}
\end{enumerate}

The final transformation is the multiplication of all the {\em local} 
transformations which have been constructed above.
The number of local transformations is bounded by $O(n)$.
Thus, the algorithm performs at most $O(n^{1+\omega}\log k + nk^3)$ ring operations.
\end{proof}

\bibliographystyle{alpha}
\bibliography{quadraticforms.bib}

\appendix

\section{Diagonalizing a Matrix}\label{sec:BlockDiagonal}

In this section, we provide a proof of Theorem \ref{thm:BlockDiagonal}.

\paragraph{Module.} There are quadratic forms which have no associated 
lattice e.g., negative
definite quadratic forms. To work with these, we define the concept of
free modules (henceforth, called module) which behave as vector 
space but have no associated realization
over the Euclidean space $\bbR^n$.

If $M$ is finitely generated $\Ring$-module with generating set
$\Vx_1,\cdots,\Vx_n$ then the elements $\Vx \in M$ can
be represented as $\sum_{i=1}^n r_i \Vx_i$, such that
$r_i \in \Ring$ for every $i \in [n]$. By construction,
for all
$a,b \in R$, and $\Vx,\Vy \in M$;
\[
a(\Vx+\Vy)=a\Vx+a\Vy \qquad (a+b)\Vx=a\Vx+b\Vx \qquad a(b\Vx)=(ab)\Vx
\qquad 1\Vx=\Vx
\]
Note that, if we replace $\Ring$ by a field in the definition 
then we get a vector space (instead of a module). 
Any inner product
$\beta:M\times M \to \Ring$ gives rise to a quadratic form 
$\MQ\in\Ring^{n\times n}$ as follows;
\[
\MQ_{ij} = \beta(\Vx_i,\Vx_j) \;.
\]
Conversely, if $R=\bbZ$ then by definition, every symmetric matrix 
$\MQ \in \bbZ^{n\times n}$ gives rise to an inner product $\beta$ 
over every $\bbZ$-module $M$; as follows.
Given $n$-ary integral quadratic form $\MQ$ and a $\bbZ$-module
$M$ generated by the basis $\{\Vx_1,\cdots,\Vx_n\}$ we define the 
corresponding inner product $\beta:M\times M \to \bbZ$ as;
\[
\beta(\Vx,\Vy)=\sum_{i,j}c_id_j\MQ_{ij}
\text{ where, }\Vx=\sum_{i}c_i\Vx_i ~~ \Vy=\sum_{j}d_j\Vx_j\;.
\]
In particular, any integral quadratic form $\MQ^n$ can be interpreted 
as describing an inner product over a free module of dimension $n$.

For studying quadratic forms over $\zpkz$, where $p$ is a prime and $k$
is a positive integer; the first step is to find equivalent quadratic 
forms which have as few mixed terms as possible (mixed terms are terms
like $x_1x_2$).

\begin{proof}(Theorem \ref{thm:BlockDiagonal})
The transformation of the matrix $\MQ$ to a block diagonal form involves
three different kinds of transformation. We first describe these 
transformations on $\MQ$ with small dimensions (2 and~3).

\begin{enumerate}[(1)]
\item Let $\MQ$ be a $2\times 2$ integral quadratic form. Let us also 
assume that
the entry with smallest $p$-order in $\MQ$ is a diagonal entry, say
$\MQ_{11}$. Then, $\MQ$ is of the following form; where $\alpha_1,\alpha_2$
and $\alpha_3$ are units of $\zpz$.
\[
\MQ = \begin{pmatrix}p^i \alpha_1 & p^j \alpha_2 \\ p^j \alpha_2 & p^s\alpha_3 
\end{pmatrix} \qquad i \leq j,s
\]
The corresponding $\MU \in \text{SL}_2(\zpkz)$, that diagonalizes $\MQ$ 
is given below. The number
$\alpha_1$ is a unit of $\zpz$ and so $\alpha_1$ has an inverse in $\zpkz$.  
\[
\MU = \begin{pmatrix} 1 & -\frac{p^{j-i}\alpha_2}{\alpha_1} \bmod{p^k} \\ 
0 & 1\end{pmatrix} 
\qquad
\MU'\MQ\MU \equiv \begin{pmatrix} p^i\alpha_1 & 0 \\ 0 & p^s\alpha_3 - 
p^{2j-i}\frac{\alpha_2^2}{\alpha_1}\end{pmatrix}\pmod{p^k}
\]

\item If $\MQ^2$ does not satisfy the condition of item (1) i.e., 
the off diagonal entry is the one with smallest $p$-order, then we start
by the following transformation $\MV \in \SL_2(\zpkz)$.
\begin{gather*}
\MV = \begin{pmatrix}1 & 0 \\ 1 & 1\end{pmatrix} \qquad 
\MV'\MQ\MV = \begin{pmatrix} \MQ_{11}+2\MQ_{12}+\MQ_{22} & \MQ_{12}+\MQ_{22} \\
\MQ_{12}+\MQ_{22} & \MQ_{22}\end{pmatrix} 
\end{gather*}
If $p$ is an odd prime then $\ordp(\MQ_{11}+2\MQ_{12}+\MQ_{22})=\ordp(\MQ_{12})$, 
because $\ordp(\MQ_{11}),$ $\ordp(\MQ_{22}) >\ordp(\MQ_{12})$. By definition,
$\MS=\MV'\MQ\MV$ is equivalent to $\MQ$ over the ring $\zpkz$. But now, $\MS$
has the property that $\ordp(\MS_{11}) = \ordp(\MS_{12})$, and it can be 
diagonalized using the transformation in (1). The final transformation
in this case is the product of $\MV$ and the subsequent transformation
from item (1). The product of two matrices from $\SL_2(\zpkz)$ is also
in $\SL_2(\zpkz)$, completing the diagonalization in this case.

\item If $p=2$, then the transformation in item (2) fails. In this case,
it is possible to subtract a linear combination of these two rows/columns
to make everything else on the same row/column equal to zero over $\ztkz$.
The simplest such transformation is in dimension~3. The situation is as 
follows. Let $\MQ^3$ be a quadratic form whose off diagonal entry has the 
lowest possible power of
$2$, say $2^{\ell}$ and all diagonal entries are divisible by at least
$2^{\ell+1}$. In this case, the matrix $\MQ$ is of the
following form.
\[
\MQ = \begin{pmatrix}2^{\ell + 1} a & 2^{\ell}b & 2^id \\
2^{\ell}b & 2^{\ell+1}c & 2^je \\
2^id & 2^je & 2^{\ell+1}f \end{pmatrix} \qquad b \text{ odd}, \ell\leq i,j
\]
In such a situation, we consider the matrix $\MU \in \SL_3(\ztkz)$
of the form below such that if $\MS=\MU'\MQ\MU \pmod{2^k}$ then 
$\MS_{13}=\MS_{23}=0$.
\begin{gather*}
\MU = \begin{pmatrix}1 & 0 & -r \\ 0 & 1 & -s \\ 0 & 0 & 1\end{pmatrix}\\
(\MU'\MQ\MU)_{13} \equiv 0 \pmod{2^k} \implies 
r 2a + s b \equiv  2^{i-\ell}d \pmod{2^{k-\ell}}\\
(\MU'\MQ\MU)_{23} \equiv 0 \modtk \implies 
r b + s 2c \equiv 2^{j-\ell}e  \pmod{2^{k-\ell}}
\end{gather*}
For $i,j \geq \ell$ and $b$ odd, the solution $r$ and $s$ can be found by 
the Cramer's rule, as below. The solutions exist because the matrix 
$\begin{pmatrix}2a & b \\ b & 2c\end{pmatrix}$ has determinant $4ac-b^2$, which
is odd and hence invertible over the ring $\bbZ/2^{k-\ell}\bbZ$.
\[
r = \frac{\det\begin{pmatrix}2^{i-\ell}d & s \\ 2^{j-\ell}e & 2c\end{pmatrix}}
{\det\begin{pmatrix}2a & b \\ b & 2c\end{pmatrix}} \pmod{2^{k-\ell}} ~~
s = \frac{\det\begin{pmatrix}2a & 2^{i-\ell}d \\ b & 2^{j-\ell}e\end{pmatrix}}
{\det\begin{pmatrix}2a & b \\ b & 2c\end{pmatrix}}  \pmod{2^{k-\ell}}
\]
\end{enumerate}

This completes the description of all the transformations we are going
to use, albeit for $n$-dimensional $\MQ$ they will be a bit technical.
The full proof for the case of odd prime follows.

Our proof will be a reduction of the problem of diagonalization from
$n$ dimensions to $(n-1)$-dimensions, for the odd primes $p$. We now
describe the reduction.

Given the matrix $\MQ^n$, let $M$ be the corresponding $(\zpkz)$-module
with basis $\MB=[\Vb_1,\cdots,\Vb_n]$ i.e., $\MQ=\MB'\MB$. We first find 
a matrix entry with the smallest $p$-order, say $\MQ_{i^*j^*}$. The 
reduction has two cases: (i) there is a diagonal entry in $\MQ$ with
the smallest $p$-order, and (ii) the smallest $p$-order occurs on an
off-diagonal entry.

We handle case (i) first. Suppose it is possible to pick $\MQ_{ii}$
as the entry with the smallest $p$-order. Our first transformation
$\MU_1 \in \sln(\zpkz)$ is the one which makes the following 
transformation i.e., swaps $\Vb_1$ and $\Vb_i$.
\begin{align}\label{BlockDiagonal:U1}
[\Vb_1,\cdots,\Vb_n] \underset{\MU_1, p^k}{\to} [\Vb_i,\Vb_2,\cdots,
\Vb_{i-1},\Vb_1,\Vb_{i+1},\cdots, \Vb_n]
\end{align}

Let us call the new set of elements $\MB_1=[\Vv_1,\cdots,\Vv_n]$ and
the new quadratic form $\MQ_1=\MB_1'\MB_1 \bmod{p^k}$. Then,
$\Vv_1'\Vv_1$ has the smallest $p$-order in $\MQ_1$ and
$\MU_1'\MQ\MU_1\equiv \MQ_1 \bmod{p^k}$. The next transformation
$\MU_2 \in \sln(\zpkz)$ is as follows. 
\begin{equation}\label{BlockDiagonal:U2}
\Vw_i = \left\{ \begin{array}{ll}
\Vv_1 & \text{if $i=1$}\\
\Vv_i - \frac{\Vv_1'\Vv_i}{p^{\ordp((\MQ_1)_{11})}} \cdot 
\left(\frac{1}{\copp((\MQ_1)_{11})} \bmod{
p^k}\right) \cdot \Vv_1 & \text{otherwise\,.}
\end{array}\right.
\end{equation}
By assumption, $(\MQ_1)_{11}$ is the matrix entry with 
the smallest $p$-order and so $p^{\ordp((\MQ_1)_{11})}$ divides 
$\Vv_{1}'\Vv_i$. Furthermore, $\copp((\MQ_1)_{11})$ is invertible 
modulo $p^k$. Thus, the transformation in Equation 
\ref{BlockDiagonal:U2} is well defined. Also note that it is
a basis transformation, which maps one basis of $\MB_1=[\Vv_1,\cdots,\Vv_n]$
to another basis $\MB_2=[\Vw_1,\cdots,\Vw_n]$. Thus, the 
corresponding basis transformation $\MU_2$ is a 
unimodular matrix over integers, and so $\MU_2\in\sln(\zpkz)$.
Let $\MQ_2=\MU_2'\MQ_1\MU_2 \bmod{p^k}$. Then, we show that the
non-diagonal entries in the entire first row and first column of
$\MQ_2$ are~0. 
\begin{align*}
(\MQ_2&)_{1i(\neq 1)}=(\MQ_2)_{{i1}}=\Vw_1'\Vw_i \bmod{p^k}\\
	&\overset{(\ref{BlockDiagonal:U2})}{\equiv} 
	\Vv_1'\Vv_i - \frac{\Vv_1'\Vv_i}{p^{\ordp((\MQ_1)_{11})}} \cdot 
	\left(\frac{1}{\copp((\MQ_1)_{11})} \bmod{p^k}\right) 
	\cdot \Vv_1'\Vv_1 \\
	&\equiv\Vv_1'\Vv_i - \frac{\Vv_1'\Vv_i}{p^{\ordp((\MQ_1)_{11})}} \cdot 
	\left(\frac{1}{\copp((\MQ_1)_{11})} \bmod{p^k}\right) 
	\cdot p^{\ordp((\MQ_1)_{11})}\copp((\MQ_1)_{11}) \\
	&\equiv 0 \bmod{p^k} 
\end{align*}

Thus, we have reduced the problem to $(n-1)$-dimensions.
We now recursively call this algorithm with the quadratic form
$\MS=[\Vw_2,\cdots,\Vw_{n}]'[\Vw_2,\cdots,\Vw_{n}] \bmod{p^k}$
and let $\MV \in \SL_{n-1}(\zpkz)$ be the output of
the recursion. Then, $\MV'\MS\MV \bmod{p^k}$ is a diagonal
matrix. Also, by consruction $\MQ_2=\diag((\MQ_2)_{11},\MS)$.
Let $\MU_3=1\oplus\MV$, and $\MU=\MU_1\MU_2\MU_3$, then,
by construction, $\MU'\MQ\MU \bmod{p^k}$ is a diagonal
matrix; as follows.
\begin{gather*}
\MU'\MQ\MU \equiv \MU_3'\MU_2'\MU_1'\MQ\MU_1\MU_2\MU_3\equiv
\MU_3'\MQ_2\MU_3\equiv(1\oplus\MV)'\diag((\MQ_2)_{11})(1\oplus\MV)\\
\equiv \diag((\MQ_2)_{11},\MV'\MS\MV) \bmod{p^k}
\end{gather*}

Otherwise, we are in case (ii) i.e., the entry with smallest 
$p$-order in $\MQ$ is an off diagonal entry, say $\MQ_{i^*j^*},
i^*\neq j^*$. Then, we make the following basis transformation
from $[\Vb_1,\cdots,\Vb_n]$ to $[\Vv_1,\cdots,\Vv_n]$ as follows.
\begin{equation}\label{BlockDiagonal:U0}
\Vv_i = \left\{ \begin{array}{ll}
\Vb_{i^*}+\Vb_{j^*} & \text{if $i=i^*$}\\
\Vb_i & \text{otherwise\,.}
\end{array}\right.
\end{equation}
The transformation matrix $\MU_0$ is from $\sln(\zpkz)$.
Recall, 
$\ordp(\MQ_{i^*j^*}) < \ordp(\MQ_{i^*i^*}), \ordp(\MQ_{j^*j^*})$, and so
$\ordp(\Vv_{i^*}'\Vv_{i^*})=\ordp(\Vb_{i^*}'\Vb_{j^*})$. Furthermore, 
$\ordp(\Vv_i'\Vv_j)\geq \ordp(\Vb_{i^*}'\Vb_{j^*})$, and so the minimum 
$p$-order does not change after the transformation in Equation 
(\ref{BlockDiagonal:U0}). This transformation reduces the problem to the 
case when the matrix entry with minimum $p$-order appears on the 
diagonal. This completes the proof of the theorem for odd primes
$p$.

For $p=2$, exactly the same set of transformations works, unless the 
situation in item (3) arises. In such a case, we use the type II block
to eliminate all other entries on the same rows/columns as the type II
block. Thus, in this case, the problem reduces to one in dimension $(n-2)$.

The algorithm uses $n$ iterations, reducing the dimension by~1 in each 
iteration. In each iteration, we have to find the minimum $p$-order, costing
$O(n^2\log k)$ ring operations and then~3 matrix multiplications costing $O(n^3)$ operations
over $\zpkz$. Thus, the overall complexity is $O(n^4+n^3\log k)$ or 
$O(n^4\log k)$ ring operations.
\end{proof}

\section{Missing Proofs}\label{sec:Proofs}

\begin{proof}(proof of Lemma \ref{lem:Square})
We split the proof in two parts: for odd primes $p$ and for the prime~2.
\begin{description}
\item [Odd Prime.]
If $0 \neq t \in \zpkz$ then $\ordp(t)<k$. If
$t$ is a square modulo $p^k$ then there exists a $x$
such that $x^2\equiv t \pmod{p^k}$. Thus, there exists $a \in \bbZ$
such that $x^2=t+ap^k$. But then, $2\ordp(x)=\ordp(t+ap^k)=\ordp(t)$.
This implies that $\ordp(t)$ is even and $\ordp(x)=\ordp(t)/2$.
Substituting this into $x^2=t+ap^k$ and dividing the entire equation
by $p^{\ordp(t)}$ yields that $\copp(t)$ is a
quadratic residue modulo $p$; as follows.
\[
\copp(x)^2 = \copp(t)+ap^{k-\ordp(t)}\equiv \copp(t) \pmod{p}
\]

Conversely, if $\copp(t)$ is a quadratic residue modulo $p$ 
then there exists a $u \in \zpkz$ such that $u^2\equiv \copp(t) \pmod{p^k}$,
by Lemma \ref{lem:Square}. If $\ordp(t)$ is even then 
$x=p^{\ordp(t)/2}u$ is a solution to the equation 
$x^2\equiv t \pmod{p^k}$.

\item [Prime $2$.] 
If $0 \neq t \in \ztkz$ then $\ordt(t)<k$. If $t$ is a square modulo 
$2^k$ then there exists an integer $x$ such that $x^2\equiv t \modtk$.
Thus, there exists an integer $a$ such that $x^2=t+a2^k$. But then,
$2\ordt(x)=\ordt(t+a2^k)=\ordt(t)$. This implies that $\ordt(t)$ is 
even and $\ordt(x)=\ordt(t)/2$. Substituting this into the equation
$x^2=t+a2^k$ and dividing the entire equation by $2^{\ordt(t)}$ 
yields,
\[
\copt(x)^2 = \copt(t) + a2^{k-\ordt(t)} \qquad \copt(t) < 2^{k-\ordt(t)}\;.
\]
But $\copt(x)$ is odd and hence
$\copt(x)^2 \equiv 1 \pmod{8}$. If $k-\ordt(t)>2$, then 
$\copt(t)\equiv 1 \pmod{8}$. Otherwise, if $k-\ordt(t)\leq 2$ 
then $\copt(t) < 2^{k-\ordp(t)}$ implies that $\copt(t)=1$.

Conversely, if $\copt(t)\equiv 1 \pmod{8}$ then there exists a 
$u \in \ztkz$ such that $u^2\equiv \copt(t) \modtk$, by Lemma
\ref{lem:Square}. If $\ordt(t)$ is even then $x=2^{\ordt(t)/2}u$
is a solution to the equation $x^2\equiv t\modtk$.
\end{description}
\end{proof}

\end{document}